\documentclass[%
 reprint,
 amsmath,amssymb,
 aps,
pra,
]{revtex4-2}
\usepackage{graphicx} 
\usepackage{tikz,pgfplots}
\pgfplotsset{compat=1.18}
\usepackage{amsmath}
\usepackage{amsfonts}
\usepackage{url}

\newcommand{\abs}[1]{\left\lvert{#1}\right\rvert}
\newcommand{\ket}[1]{\left\lvert{#1}\right\rangle}
\newcommand{\bra}[1]{\left\langle{#1}\right\rvert}
\newcommand{\ketbra}[2]{\left\lvert{#1}\middle\rangle\middle\langle{#2}\right\rvert}
\newcommand{\norm}[1]{\left\lvert\left\lvert {#1} \right\rvert\right\rvert}
\newcommand{\rank}[1]{\operatorname{rank}{\left(#1\right)}}
\newcommand{\braket}[2]{\left\langle{#1}\middle\vert{#2}\right\rangle}

\newcommand{\comment}[1]{}

\usepackage{xcolor}

\DeclareMathOperator{\tr}{tr}

\usepackage{amsthm}
\newtheorem{lem}{Lemma}
\newtheorem{thm}{Theorem}
\newtheorem{cor}{Corollary}
\newtheorem{con}{Conjecture}

\usepackage{placeins} 
\usepackage{hyperref}

\begin{document}
\preprint{APS/123-QED}
\title{Exponential lower bounds on the fermionic Gaussian rank of magic states and the bosonic coherent state rank of Fock states}
\author{Oliver Reardon-Smith }
\affiliation{Centrum Fizyki Teoretycznej Polskiej Akademii Nauk} 
\date{\today}
\begin{abstract}

Recent algorithms for classical simulation of quantum mechanics have runtime whose superpolynomial component is given by a linear dependence on the number of terms required to write large tensor products of certain ``magic''  states as superpositions of ``free'' states (which may be stabilizer states, fermionic Gaussian states or others). Surprisingly little is known about the number of terms in such decompositions, called \emph{ranks} (e.g. the stabilizer rank, fermionic-Gaussian rank etc.). For complexity theoretic reasons they are expected to grow exponentially in the number of tensor factors but, while exponential upper bounds are known for both the stabilizer and fermionic-Gaussian rank of magic states; the best known lower bounds on the stabiliser rank are quadratic, and no bounds on the fermionic-Gaussian rank are known beyond fixed constants. In this work we prove that any fermionic Gaussian decomposition of $\ket{M}^{\otimes k}$ consists of $\Omega(1.4^k)$ terms. Here $\ket{M}$ is the most standard magic state for fermionic linear optics: the $4$ qubit state that may be consumed to implement a swap gate. We also prove essentially matching bounds on the $\delta$-approximate rank of the same state, lower bounding it by the same quantity that bounds the exact rank, multiplied by a factor of $1-\delta^2$. Our results on exact fermionic Gaussian rank apply directly to any tensor product of $k$ fixed parity non-Gaussian states, although the same is not true of the approximate rank. In the bosonic context we prove that the coherent state border rank of an $n$-mode bosonic Fock state with $m_j$ bosons in mode $j$ is exactly $\prod_{j}(1+m_j)$, answering as a special case a conjecture of Ref.~\cite{cottier2026lowerboundscoherentstate} and obtain lower bounds on the approximate coherent state rank given by the same quantity multiplied by a function of the fidelity of the approximation, emphasizing the broad applicability of the method we employ.

\end{abstract}
\maketitle

\section{Introduction}
\label{sec:intro}
In recent years, significant attention has been focused on a class of classical simulation algorithms for quantum circuits based on decompositions of magic states. This paradigm has been successfully applied to circuits formed of Clifford+$T$ gates~\cite{bravyi2016improved,bravyi-2019,PRXQuantum.3.020361,Kissinger-2022} as well as matchgates \footnote{Qubit quantum circuits consisting of matchgates are exactly equivalent to so-called fermionic linear optics due to the results of~\cite{PhysRevA.65.032325} and we will use the two terms interchangeably.} augmented by swap gates~\cite{Dias_2024,ReardonSmith2024improvedsimulation}. In each case the resulting simulation algorithms have polynomial runtime in the number of qubits, and number of ``free'' gates (in these examples the Clifford and matchgates, respectively) with the non-polynomial component of the runtime being controlled by the non-free gates. Due to a now standard trick known as gate teleportation or gadgetization, the non-free gates in a circuit can be replaced by a combination of free gates and non-free magic states~\cite{gottesman1999demonstrating,zhou2000methodology,Yoganathan-2019-magic-for-matchgates,PhysRevA.73.042313}. The non-polynomial component of the classical simulation algorithms is then given by the number of terms required to decompose a tensor product of $k$ magic states as a linear combination of free states, a quantity known as the rank. For many simulation tasks an exact decomposition is not required, and the practically relevant quantity is the $\delta$-approximate rank, defined as the least rank among states with $L^2$ distance at most $\delta$ from the target state. 

This approach to classical simulation has also been applied to the continuous variable setting of bosonic linear optics. Classical simulation algorithms based on expressing statevectors as a linear combination of bosonic Gaussian states were developed in Refs.~\cite{PhysRevA.110.042402,Hahn2025classicalsimulation}. Recently Ref.~\cite{cottier2026lowerboundscoherentstate} has proved superpolynomial, but less than exponential, lower bounds on the coherent-state (border) rank of the Fock state $\ket{1}^{\otimes n}$ in bosonic linear optics and has conjectured that this rank is exactly $2^n$. We prove this conjecture (and a broader result for arbitrary Fock states) in theorem~\ref{thm:coherent-state-rank-for-fock-states}. The authors of Ref.~\cite{cottier2026lowerboundscoherentstate} note that coherent states are a strict subset of the class of bosonic Gaussian states, for which classical simulation is efficient (assuming the evolution and measurements are also Gaussian) and leave open the question of bounding the rank over the full set of Gaussian states, as do we.

 For complexity-theoretic reasons, the rank and approximate rank of both tensor products of stabilizer and fermionic Gaussian magic states (with respect to their respective classes of free states) is expected to grow exponentially with the number of terms in the tensor product. A polynomial rank (or approximate rank, with exponentially small $\delta$) in the stabilizer setting is known to imply that the $\mathrm{P}^{\mathrm{\#P}} = \mathrm{MA}$~\cite{mehraban2024quadraticlowerboundsapproximate}, while an effective \footnote{Here, by effective we mean that e.g. high complexity is not hidden in the coefficients of the decomposition. A priori one could have even a constant rank, but coefficients that are exponentially difficult to compute.} constructive polynomial decomposition of stabilizer magic states would imply an efficient solution of $3$-$\mathrm{SAT}$ problems, contradicting the strong exponential time hypothesis and proving $\mathrm{P}= \mathrm{NP}$~\cite{huang2018explicitlowerboundsstrong}. On the matchgate side, a polynomial fermionic-Gaussian decomposition of fermionic magic states would lead to efficient classical computation of probabilities of fermion sampling~\cite{oszmaniec2022fermion} and, under the same caveats about computationally dealing with numerical constants as in the Clifford+$T$ case, would also imply polynomial solutions to $\#\mathrm{P}$-hard problems~\cite{hebenstreit2020computational}.

Despite these strong complexity theoretic results, and the importance of the ranks for classical simulation, surprisingly little is known about them without assuming unproven complexity theoretic hypotheses. In both cases exponential upper bounds have been proved by explicit decompositions, see Refs.~\cite{bravyi2016trading,bravyi2016improved,bravyi-2019,Qassim2021improvedupperbounds,Kissinger-2022} for the stabilizer case and the trivial computational basis decomposition for the matchgate case. Turning to lower bounds, the state of the art for the stabilizer rank of $\ket{T}^{\otimes k}$ is $\tilde{\Omega}(k^2)$, up to logarithmic factors~\cite{bravyi2016trading, Peleg2022lowerbounds,Lovitz2022newtechniques,mehraban2024quadraticlowerboundsapproximate}. While not directly relevant to magic states, Ref.~\cite{Lovitz2022newtechniques} also proves the exponential stabilizer rank of a specific tensor-product state, this state is constructed with a exponentially biased coefficients in the computational basis and, in particular, has $\delta$-approximate rank bounded by a constant. 

The situation for the fermionic Gaussian rank of $\ket{M}^{\otimes k}$ is even worse, no bounds have been proved beyond the trivial fact that since $\ket{M}^{\otimes k}$ is non-Gaussian for $k\geq 2$ the rank must be at least $2$, although for $k=2$ and $k=3$ both the symmetry restricted and numerical results of Ref.~\cite{cudby2025gaussiandecompositionmagicstates} make a rank below the trivially attainable value of $2^k$ unlikely. 

In this work we partially remedy this situation and provide unconditional exponential lower bounds on the fermionic Gaussian rank and approximate fermionic Gaussian rank of a tensor product of $k$ copies of the magic state we denote as
\begin{align}
    \ket{M} &= \frac{1}{\sqrt{2}}\left(\ket{0000} + \ket{1111}\right),
\end{align}
this state, combined with matchgate circuits and either classical control or post-selection is known to be able to implement a swap gate, and hence to promote matchgate circuits to quantum universality~\cite{PhysRevA.73.042313,josaMatchgate2008,Yoganathan-2019-magic-for-matchgates}. While exponential, the asymptotic bounds are too weak to imply non-trivial ranks for small numbers of copies of $\ket{M}$ states (for example they do not prove that a tensor product of $k=2$ or $k=3$ copies has a rank larger than $1$), which we address with our secondary result, that the gaussian rank of $\ket{M}^{\otimes 2}$ is exactly $4$ implying, since we can delete a magic state by Gaussian post-processing, that the gaussian rank for $\ket{M}^{\otimes k}$ where $k\geq 2$ is not less than $4$. We summarize our fermionic results in theorems~\ref{thm:exact-rank}-\ref{thm:rank-two-copies} and corollary~\ref{cor:exact-rank-general-fixed-parity-4-mode} below.
\begin{thm}[Exact rank]\label{thm:exact-rank}
    If $\ket{M} = \frac{1}{\sqrt{2}}(\ket{0000} + \ket{1111})$ and 
    \begin{align}
    \ket{M}^{\otimes k} = \sum_{i=1}^r \alpha_i \ket{f_i},
    \end{align}
    for some fermionic Gaussian states $\ket{f_i}$ then
    \begin{align}
    r &\geq \frac{8}{9(k+1)} \left(\frac{27(8+2\sqrt{7})}{256}\right)^k \\&> \frac{8}{9(k+1)} 1.4018^k.
    \end{align}
\end{thm}

\begin{thm}[$\delta$-approximate rank]\label{thm:approx-rank}
    If $\ket{M} = \frac{1}{\sqrt{2}}(\ket{0000} + \ket{1111})$ and 
    \begin{align}
    \norm{\ket{M}^{\otimes k}  - \sum_{i=1}^r \alpha_i \ket{f_i}} < \delta,
    \end{align}
    for some fermionic Gaussian states $\ket{f_i}$ and real $\delta$ then
    \begin{align}
    r &\geq (1-\delta^2)\frac{8}{9(k+1)} \left(\frac{27(8+2\sqrt{7})}{256}\right)^k\\ &> (1-\delta^2) \frac{8}{9(k+1)} 1.4018^k.
    \end{align}
\end{thm}

\begin{thm}[Rank for $k = 2$ copies]\label{thm:rank-two-copies}
    If 
    \begin{align}
    \ket{M}^{\otimes 2} = \sum_{i=1}^r \alpha_i \ket{f_i},
    \end{align}
    for some fermionic Gaussian states $\ket{f_i}$ then
    \begin{align}
    r \geq 4.
    \end{align}
\end{thm}
Combining theorems~\ref{thm:exact-rank} and~\ref{thm:rank-two-copies} with lemma~\ref{lem:arbitrary-tensor-product} we obtain Corollary~\ref{cor:exact-rank-general-fixed-parity-4-mode}.
\begin{cor}
    \label{cor:exact-rank-general-fixed-parity-4-mode}
    If
    \begin{align}
    \bigotimes_{j=1}^n \ket{\varphi_j} = \sum_{i=1}^r \alpha_i \ket{f_i},
    \end{align}
    for some fermionic Gaussian states $\ket{f_i}$, and $k >1$ of the $\ket{\varphi_j}$ are non-Gaussian states of fixed parity, then
    \begin{align}
    r &\geq \max\left\{\frac{8}{9(k+1)} 1.4018^k, \, 4\right\}.
    \end{align}
\end{cor}
While the lower bound we prove for magic states in theorem~\ref{thm:exact-rank} is ``not too bad'' in the sense that the base of the exponential factor $\approx 1.4$ is reasonably close to the known upper bound of $2$, we expect the bound on tensor products of more general states from corollary~\ref{cor:exact-rank-general-fixed-parity-4-mode} to be extremely loose. When applying the method of Ref.~\cite{Yoganathan-2019-magic-for-matchgates}to map non-Gaussian states of more than $4$ qubits to non-Gaussian $4$ qubit state we expect to lose non-Gaussianity.

We summarise the bounds on the exact rank of theorems~\ref{thm:exact-rank} and~\ref{thm:rank-two-copies} in figure~\ref{fig:exact-rank-bounds}, which also contains the data from table~\ref{tab:table-of-exact-values}, showing lower bounds for small $k$ obtained by directly evaluating equation~\eqref{eqn:lower-bound-from-matrix-ranks} and optimizing over the choice of parameter $p$. This numerical optimization can be formulated with integer arithmetic only, so the bound plotted is a true lower bound, not a numerical approximation, despite not being from a closed formula.

\begin{figure}[t]
	\begin{tikzpicture}
	\begin{semilogyaxis}[
		width=0.45\textwidth,
		height=0.4\textwidth,
		xlabel=$k$,
		ylabel={Rank of $\ket{M}^{\otimes k}$},
		xmin=1, xmax=30,
        grid=major,legend style={at={(0.05,0.95)},anchor=north west}]
		\addplot[color=orange,mark=*,mark size=2pt] table [x index=0,y index=1,col sep=space] {plot_data.txt};
		\addlegendentry{Lower bound theorem~\ref{thm:exact-rank}}
		\addplot[color=blue,dashed] table [x index=0,y index=2,col sep=space] {plot_data.txt};
		\addlegendentry{Lower bound theorem~\ref{thm:rank-two-copies}}
        \addplot[color=red,mark=triangle*,mark size=2pt] table [x index=0,y index=3,col sep=space] {plot_data.txt};
        \addlegendentry{$2^k$ upper bound}
        \addplot[color=cyan,mark=x,mark size=2pt] table [x index=0,y index=4,col sep=space] {plot_data.txt};
        \addlegendentry{Optimized lower bound table~\ref{tab:table-of-exact-values}}
	\end{semilogyaxis}
	\end{tikzpicture}
	\caption{A summary of our results and the now known bounds on the exact fermionic Gaussian rank of the state $\ket{M}^{\otimes k}$. The trivial upper bound and the lower bound from theorem~\ref{thm:exact-rank} apply to and are plotted from $k=1$, whereas the bound from theorem~\ref{thm:rank-two-copies} applies only to $k\geq 2$. We improve on the prior state of the art knowledge for every $k>1$, since the rank for $k=1$ is trivially $2$. The ``optimized'' lower bound corresponds to directly evaluating the matrix ranks on the right hand side of equation~\eqref{eqn:lower-bound-from-matrix-ranks}. The script producing the data is available at Ref.~\cite{secant-computation-github} and the exact numbers are in table~\ref{tab:table-of-exact-values}.
    \label{fig:exact-rank-bounds}}
\end{figure}
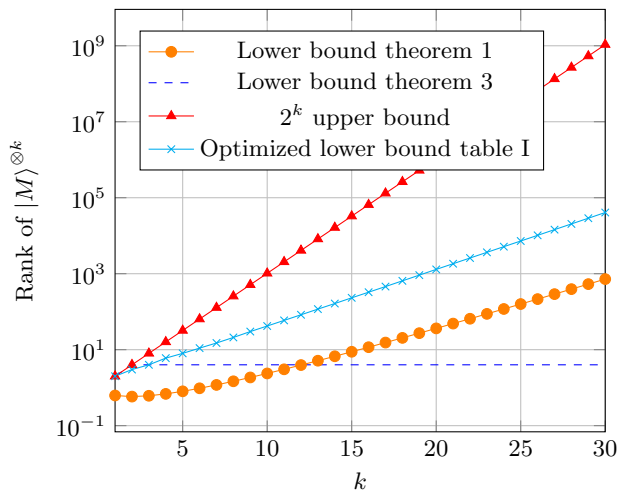

\textbf{Outline:} 
After defining notation and explaining preliminary information in section~\ref{sec:prelims} we prove theorem~\ref{thm:exact-rank} bounding the exact Gaussian rank of a tensor product of $k$ fermionic magic states in section~\ref{sec:exact-rank}, in terms of a fraction whose numerator $N_k$ we bound in subsection~\ref{sec:numerator} and denominator $D_k$ we bound in subsection~\ref{sec:denominator}. In section~\ref{sec:approximate-rank} we prove theorem~\ref{thm:approx-rank}, showing that essentially the same bound holds for the $\delta$-approximate rank, up to a correction factor of $(1-\delta^2)$. Theorem~\ref{thm:rank-two-copies} is proved in section~\ref{sec:two-copies-only}. We discuss conclusions and the outlook for future work section~\ref{sec:conclusion}. 

While finalizing this manuscript the author became aware of the independent work Ref.~\cite{wei2026robustexponentiallowerbounds}, in which  similar methods are used to obtain results essentially matching our asymptotic results in the fermionic case, as well as bounds on the bosonic Gaussian border rank. Ref.~\cite{wei2026robustexponentiallowerbounds} does not appear to prove results equivalent to those appearing here on the bosonic coherent state rank of Fock states or on the exact fermionic rank for $2$ copies of the $\ket{M}$ state.

\section{Preliminaries and notation}
\label{sec:prelims}
Let $\mathcal{H} = \mathbb{C}_2^{\otimes n}$ be the Hilbert space of an $n$ qubit system. This system is equipped with the well-known group of $4^{n+1}$ Pauli operators with phases in $\{1,i,-1,-i\}$. Clifford operators are the defined as the normalizer of the Pauli group within the unitary group (i.e. unitaries which map products of Paulis to products of Paulis). We will consider only pure states and abuse notation slightly, saying quantum states are given by vectors in $\mathcal{H}$ (ignoring phase-equivalence). A stabilizer state is a quantum state which is a simultaneous eigenvector of $2^n$ commuting elements of the Pauli group, with each eigenvalue equal to $+1$. The Pauli operators for which the state is a $+1$-eigenvector is called it's stabilizer group. A standard result is that stabilizer states are exactly those which may be obtained from a computational basis state by applying a Clifford operator. A projector onto a stabilizer state may be written as an average over elements of its stabilizer group, or a product over a set of $n$ independent non-identity stabilizer group elements
\begin{align}
    \ketbra{\varphi}{\varphi} = 2^{-n} \sum_{g\in G}g=  \prod_{j=1}^n\frac{1}{2}(I + g_j).
\end{align}
The Majorana fermion operators $m_i$, for $i=1\hdots 2n$ may be used as an alternative generating set for the $n$ qubit Pauli group (with appropriate phases). In terms of the Pauli operators we may write
\begin{align}
    m_{2i-1} &= Z^{\otimes(i-1)} \otimes Y  \otimes I^{\otimes(n-i)}\\
    m_{2i} &= Z^{\otimes(i-1)} \otimes X  \otimes I^{\otimes(n-i)}.
\end{align}
The Majorana fermion operators obey a version of the canonical anti-commutation relations
\begin{align}
    \{m_i, m_j\} &= 2I \delta_{ij},
\end{align}
and are related to fermion physics by the Jordan-Wigner transformation, which relates their actions to operators creating or annihilating a fermion, where in qubit terms a computational basis state $\ket{x}$ corresponds to a Fock state with modes indexed by $i$ such that $x_i = 1$ occupied by a (spinless) fermion, and modes with $x_i= 0$ empty. 

Fermionic linear optics is a subtheory of quantum mechanics similar to Clifford/stabilizer quantum mechanics. The role of unitary Clifford operations is filled by fermionic Gaussian unitaries. Where the Cliffords map products of Majorana fermion operators to products of Majorana fermion operators, the Gaussian unitaries map \emph{sums} of Majorana fermion operators to sums of Majorana fermion operators. In symbols they satisfy
\begin{align}
    U m_j U^\dagger &= \sum_{k} R_{jk} m_k,
\end{align}
one can show that an $R$ satisfying this relation must be real and orthogonal. One may also show that Gaussian unitaries do not mix the eigenspaces of the parity operator, i.e. if we denote the parity operator by $P$ then for any Gaussian $U$ we have  $UP U^\dagger = \pm P$. The group of Gaussian unitaries which commute with parity is isomorphic to the well-known spin group $\mathrm{Spin}(2n)$, which is a double cover of the special orthogonal group $\mathrm{SO}(2n)$, while the full group of Gaussian unitaries is isomorphic to the group $\mathrm{Pin}(2n)$, a double cover of $\mathrm{O}(2n)$. For notational convenience we will treat these equivalences as identities and write $\mathrm{Spin}(2n)$ for the group of parity preserving Gaussian unitaries and $\mathrm{Pin}(2n)$ for the group of general Gaussian unitaries.

Fermionic Gaussian states are those which one can obtain by acting by a Gaussian unitary on an initial computational basis state. It follows that any Gaussian state is an eigenstate of the parity operator and a Gaussian decomposition of any parity eigenstate may, with no loss of generality, be assumed to consist of Gaussian states with the same parity subspace. While working within one parity subspace it suffices to consider only the spin subgroup of the full group of Gaussian unitaries.

In our use of the results of the more mathematical references~\cite{Galgano-secant-varieties, Landsberg-equations-for-secant-varieties,Holweck-Oeding-2025,Holweck-Levay-2015,Manivel-2009,Xiao-binary-forms}, we will encounter some objects which go somewhat beyond the usual presentation in the quantum information literature. The complexification of the group $\mathrm{Spin}(2n)$ will be denoted $\mathrm{Spin}(2n, \mathbb{C})$. While any parity-preserving Gaussian unitary may be written as an exponential
\begin{align}
    U = \exp\left(\frac{1}{4}\sum_{ij}\alpha_{ij}m_im_j \right) \in \mathrm{Spin}(2n) \subset U(2^n),
\end{align}
for some real, skew symmetric matrix $\alpha$, the complexification allows \emph{complex} $\alpha$. As such the  operators in the complexification are generally not unitary, but they are invertible. For our purpose it is important that the Gaussian rank of a state $\ket{\varphi}$, defined as the least $r$ such that
\begin{align}
    \ket{\varphi} &= \sum_{i=1}^r \alpha_i \ket{f_i},
\end{align}
holds for Gaussian states $\ket{f_i}$ is invariant under multiplying $\ket{\varphi}$ both by Gaussian unitaries and complexified Gaussian operators. In particular a complexified Gaussian operator preserves the set of Gaussian states. For convenience we note here lemma~\ref{lem:spin-orbits-in-4-mode-space}, showing that within each fixed parity subspace of $\mathbb{C}_2^{\otimes 4}$ there are (up to normalization, or equivalently projectively) exactly two $\mathrm{Spin}(8, \mathbb{C})$ orbits.
\begin{lem}[$\mathrm{Spin}(8, \mathbb{C})$ orbits of fixed-parity  $4$ qubit states]\label{lem:spin-orbits-in-4-mode-space}
    If $\ket{\varphi} \in \mathbb{C}_2^{\otimes 4}$ and $P \ket{\varphi} = (-1)^x \ket{\varphi}$ then there exists a complexified Guassian operator $X \in \mathrm{Spin}(8, \mathbb{C})$ such that one of the following statements is true
    \begin{enumerate}
        \item $X \ket{\varphi} \propto (m_0)^x\ket{0000}$
        \item $X \ket{\varphi} \propto (m_0)^x\ket{M}$,
    \end{enumerate}
    if both are true the statevector is the zero-vector.
\end{lem}
We prove lemma~\ref{lem:spin-orbits-in-4-mode-space} for the even parity case, since the odd case is essentially identical. Due to the results of Ref.~\cite{Oszmaniec2014} either $\ket{\varphi}$ is Gaussian, in which case we are in case 1. or $\ket{\varphi}$ may be written as a superposition of two orthogonal Gaussian states which by multiplying by an Gaussian unitary may be taken to be $\ket{0000}$ and $\ket{1111}$ and have real, positive coefficients. Now we have
\begin{align}
    U \ket{\varphi}= \alpha \ket{0000} + \sqrt{1-\alpha^2}\ket{1111},
\end{align}
for $\alpha \in \mathbb{R}^+$. Now the operator
\begin{align}
    S = \exp\left( a m_1 m_2  \right),
\end{align}
for arbitrary complex $a$ is easily seen to be within $\mathrm{Spin}(8, \mathbb{C})$, and by choosing $a$ appropriately we can map $U\ket{\varphi}$ to $\ket{M}$. We generalise lemma~\ref{lem:spin-orbits-in-4-mode-space} by also allowing more general Gaussian operators in the sense of Ref.~\cite{bravyi-2005-lagrangian-rep}, including both invertible operators in $\mathrm{Spin}(2n,\mathbb{C})$ and certain Gaussian projectors. Since these projectors map Gaussian states to Gaussian states or $0$, they can decrease the rank can not increase it. Equivalently, one can say the rank can not be increased by post-selected Gaussian protocols~\cite{tarabunga2026fermionicnongaussianitybellsampling}.
\begin{lem}\label{lem:arbitrary-tensor-product}
    Let 
    \begin{align}
        \ket{\phi} = \bigotimes_{j=1}^m \ket{\varphi_j}
    \end{align}
    and let $k$ of the $\ket{\varphi_j}$ be non-Gaussian states of fixed parity, then the exact rank of $\ket{\phi}$ is at least that of $\ket{M}^{\otimes k}$. 
\end{lem}
Lemma~\ref{lem:arbitrary-tensor-product} follows from the fact that the rank can not be increased by arbitrary Gaussian operator in the sense of Ref.~\cite{bravyi-2005-lagrangian-rep}, which includes both Gaussian unitaries and computational-basis projectors. In other words the rank is non-increasing under post-selected Gaussian protocols. We first simply measure each qubit of each $\ket{\varphi_j}$ which is \emph{not} of fixed parity in the computational basis to map these to some computational basis state. These are now fermionic Gaussian and taking the tensor product with them does not change the rank. We then apply the probabilistic protocol outlined in the proof of theorem~1 of Ref.~\cite{Yoganathan-2019-magic-for-matchgates} to bring every non-Gaussian state of fixed parity to a $4$-qubit non-Gaussian state of fixed parity, at which point we may invoke lemma~\ref{lem:spin-orbits-in-4-mode-space} to show that the tensor product of these has the same rank as $\ket{M}^{\otimes k}$. We emphasize that the comparable statement for the \emph{approximate} rank is false. For any fixed $\delta$ one can easily form a tensor product of $k$ non-Gaussian states of fixed parity whose $\delta$-approximate rank is $1$.

In our work we will actually be proving slightly stronger results that those we report concerning the exact rank. Using the notation of Ref.~\cite{Manivel-2009}, we form the positive-parity projective \emph{spinor variety} $S_+$ by taking the quotient  
\begin{align}
    S_+ &= \{\alpha U\ket{0} |\, U\in \mathrm{Spin}(2n),\, \alpha\in\mathbb{C},\, \alpha\neq 0\}/\sim
\end{align}
where $\ket{\varphi} \sim \alpha\ket{\varphi} $ for all $\alpha\neq 0$. We then form the secant variety $\sigma(S_+)$ and the second-secant variety $\sigma_2(S_+)$, respectively as the topological closure of linear combinations of $2$ and $3$ elements of $S_+$. In physics language these give the set of statevectors we can form by considering superpositions of $2$ and $3$ Gaussian states (of matching parity). Mathematically the closure here is required for the sets to form a variety, since we are focussed on lower bounds here the closure has no impact on our results, other than we prove mildly stronger results than those needed for the exact rank. Where we aim to show that a state can not be written as a superposition of $s$ Gaussian states, we in fact show it also can not be written as a limit of such superpositions.

Our main mathematical tool may be viewed as a so-called generalized flattening map~~\cite{equations-for-secant,buczyński2026cactusbarriers}. Given a set of linear maps $A_j$ and a state $\ket{\psi}$ all acting on the appropriate vector space, form the matrix whose $j^\text{th}$ column is $A_j\ket{\psi}$
\begin{align}
    M(\psi) = \left(A_1\ket{\psi}, A_2\ket{\psi},\hdots A_N\ket{\psi}\right),
\end{align}
now given a decomposition of some target state
\begin{align}
    \ket{\psi} = \sum_{j=1}^r a_j \ket{\phi_j},
\end{align}
the subaditivity of the matrix rank implies
\begin{align}
    \mathrm{rank}(M(\ket{\psi})) \leq \sum_{j=1}^r \mathrm{rank}(M(\ket{\phi_j})),
\end{align}
rearranging this inequality, and equipped with bounds on $\mathrm{rank}(M(\ket{\psi}))$ and $\mathrm{rank}(M(\ket{\phi_j}))$ one can obtain a lower bound on $r$. Up to choosing the appropriate linear maps $A_j$ to obtain the result, this is essentially the proof used for Proposition 5.1.1. in Ref.~\cite{Landsberg-equations-for-secant-varieties} and similar reasoning appears in Proposition 1.3.4 of Ref.~\cite{Steffan-phd-thesis} and in Lemma 19. of Ref.~\cite{CHRISTANDL2018125}.

\section{Exact fermionic Gaussian rank}
\label{sec:exact-rank}
Given an $n$ qubit state $\ket{\psi}$ the Gaussian rank is defined as the minimum integer $r$ such that 
\begin{align}
    \ket{\psi} = \sum_{i=1}^r \alpha_i \ket{f_i},
\end{align}
where the $\ket{f_i}$ are fermionic-Gaussian states. We aim to lower bound the rank of a tensor product of $4$-mode fermionic ``magic'' states
\begin{align}
    \ket{M} = \frac{1}{\sqrt{2}}(\ket{0000} + \ket{1111}).
\end{align}
We will employ certain matrix-valued functions of statevectors, which may be defined in terms of linear maps taking statevectors to linear maps taking degree $p$ Majorana polynomials to statevectors. For example if $\psi$ is some statevector and $x$ is an operator in the span of the Majorana monomials of degree $p$ then 
\begin{align}
    \hat{C}_p(\psi) : x \mapsto x\ket{\psi},
\end{align}
where the hat is only to discriminate between the linear map $\hat{C}_p$ and the matrix $C_p(\psi)$. Under fixing an indexing order for the Majorana monomials of degree $p$ we can express $C_p(\psi)$ as a matrix where each column is given by a particular monomial applied to $\ket{\psi}$, i.e.
\begin{align}
    C_p(\psi) = \begin{pmatrix}
        x_1 \ket{\psi}, & x_2 \ket{\psi},\hdots &,x_{\begin{pmatrix}
            2n\\p
        \end{pmatrix}} \ket{\psi}\label{eqn:c-map-def-1}
    \end{pmatrix},
\end{align}
for example if we label the $2n$ individual Majorana operators $m_i$ then
\begin{align}
    C_1(\psi) = \begin{pmatrix}
        m_1 \ket{\psi}, & m_2 \ket{\psi},\hdots &,m_{2n} \ket{\psi}
    \end{pmatrix}.\label{eqn:c-map-def-2}
\end{align}
The $C_1$ map we have defined is essentially the map $\varphi_q$ defined in equation (1.3) of Ref.~\cite{Galgano-secant-varieties}. Recall that for each $p$ the space spanned by degree-$p$ monomials is invariant under the action of the fermionic-linear optical group.

Our proof will depend on the rank of the matrix $C_p$. First note that for each $p$ every fermionic Gaussian state has the same such rank, since the span of the degree $p$ Majorana monomials is invariant under the action of the group of Gaussian unitaries. Now write
\begin{align}
    \ket{M}^{\otimes k} = \sum_{i=1}^r \ket{f_i},
\end{align}
where each $\ket{f_i}$ is an unnormalized Gaussian state. Since the matrix rank is subadditive we have 
\begin{align}
    \rank{C_p(\ket{M}^{\otimes k})} &\leq \sum_{i=1}^r \rank{ C_p(\ket{f_i})}\\ &= r \rank{C_p(\ket{0})},
\end{align}
so for each $p$ we have
\begin{align}
    r \geq \frac{\rank{C_p(\ket{M}^{\otimes k})}}{\rank{C_p(\ket{0})}} = \frac{N_p}{D_p}.\label{eqn:lower-bound-from-matrix-ranks}
\end{align}
The problem is now to lower bound the numerator and upper bound the denominator of this fraction. Any choice of $p$ will give a lower bound on the rank, although a bad choice of $p$ may lead to a trivial or uninteresting bound. For convenience and based on numerical evidence we choose here $p=k$.

In sections~\ref{sec:numerator} and~\ref{sec:denominator} below we prove the bounds
\begin{align}
    N_k &\geq \frac{1}{k+1}(1+\sqrt{7})^{2k}\\
    D_k &\leq \frac{9}{8} \left(\frac{256}{27}\right)^k,
\end{align}
combining these we have
\begin{align}
    r \geq \frac{8}{9(k+1)} \left(\frac{27(8+2\sqrt{7})}{256}\right)^k,
\end{align}
where the number being taken to the power $k$ is greater than $1.4018$.
\subsection{Numerator}
\label{sec:numerator}
In this section we aim to lower bound the quantity $\rank{C_p(\ket{M}^{\otimes k})}$. We recall that under the Jordan-Wigner transformation the Majorana fermion operators are products of Pauli operators and the magic state 
\begin{align}
    \ket{M} &= \frac{1}{\sqrt{2}}(\ket{0000} + \ket{1111}),
\end{align}
is a stabilizer state. For the calculations in this section the phases which make the $n$-qubit Pauli group non-Abelian have no effect, so we work in the projective Pauli group, which is isomorphic to the group $\mathbb{Z}_2^{2n}$ or the group of length $2n$ binary vectors. For notational simplicity we will abuse notation and write Pauli operators where we mean equivalence classes in the projective group. The stabilizer group of $\ket{M}$ is easily seen to be generated by the Pauli strings $z_1z_2$, $z_2z_3$, $z_3z_4$ and $x_1 x_2 x_3x_4$, in the basis formed by Majorana fermion operators these generators can be written in the matrix
\begin{align}
    \begin{pmatrix}
        1&1&1&1&0&0&0&0\\
        0&0&1&1&1&1&0&0\\
        0&0&0&0&1&1&1&1\\
        0&1&1&0&0&1&1&0\\
    \end{pmatrix}.\label{eqn:one-block-generating-matrix}
\end{align}
We will now consider the cosets of the projective Pauli group with respect to the subgroup formed by the stabilizer group of $\ket{M}$. Assume we have two Pauli strings $P$ and $Q$ from different cosets, then the product $PQ$ is a Pauli string which is \emph{not} proportional to any Pauli string in the stabilizer group. Since the stabilizer group is a maximal commuting group it follows that there is some stabilizer $g$ of $\ket{M}$ which anticommutes with $PQ$, and hence
\begin{align}
    \bra{M} PQ \ket{M} &= \bra{M} PQg\ket{M}\\
    &= -\bra{M} gPQ\ket{M}\\
    &= -\bra{M} PQ \ket{M}\\
    &= 0.
\end{align}
In particular we observe that $P\ket{M}$ and $Q\ket{M}$ are orthogonal. Conversely, if $P$ and $Q$ belong to the same coset then $ P\ket{M}$ and $Q\ket{M}$ are proportional to each other.

The number of orthogonal vectors obtained by applying order $p$ monomials in Majorana fermion operators to $\ket{M}$ is therefore given exactly by the number of cosets containing an element of Hamming weight (in the Majorana basis) $p$. One can check that there are exactly $16$ cosets from Lagrange's theorem. They can be described as follows
\begin{itemize}
    \item $1$ coset consisting of the stabilizer group itself, with attained Hamming weights $\{0, 4,8\}$,
    \item $8$ cosets formed by adding each of the $8$ vectors of Hamming weight $1$ to the stabilizer group, with attained Hamming weights $\{1, 3, 5,7\}$,
    \item $7$ cosets by adding $7$ linearly independent vectors of Hamming-weight 2 to the stabilizer group, with attained Hamming weights $\{2,4,6\}$. Explicitly one can take the $7$ operators $m_1 m_i$ for $i=2\hdots 8$ as representatives.
\end{itemize}
The stabilizer of a tensor-product of $t$ magic states is simply given (up to phases) by the direct sum in the projective Pauli group. All the elements in the stabilizer group are even in the Majorana fermion operators which means that they are local to each block (the trailing $Z$ operators used to ensure anti-commutation do not reach outside any block). This means that in Majorana coordinates the appropriate code generating matrix for the tensor-product is simply a direct-sum of blocks of the form shown in equation~\eqref{eqn:one-block-generating-matrix}. The cosets in the group $\mathbb{Z}_2^{8k}$ of the projective Pauli stabilizer group of $\ket{M}^{\otimes k}$ are therefore just products of cosets from each block.

For convenience say a global coset with $a$ blocks from the first bullet point above, $b$ blocks from the second bullet point and $c$ blocks from the third has type $(a,b,c)$. We must have $a+b+c = k$ and there are
\begin{align}
    \begin{pmatrix}
        k\\a
    \end{pmatrix}\begin{pmatrix}
        k-a\\b
    \end{pmatrix}8^b 7^c,
\end{align}
ways to choose such cosets. We therefore have 
\begin{align}
    \sum_{\substack{a+b \leq k\\p\in a A + bB+cC}} \begin{pmatrix}
        k\\a
    \end{pmatrix}\begin{pmatrix}
        k-a\\b
    \end{pmatrix}8^b 7^c
\end{align}
cosets containing a string of Hamming-weight $p$ where $A = \{0,4,8\}$, $B = \{1,3,5,7\}$, $C = \{2,4,6\}$. Alternatively we can write the combinatorial factor as a multinomial coefficient obtaining
\begin{align}
    \sum_{\substack{a+b \leq k\\p\in a A + bB+cC}} \begin{pmatrix}
        k\\a,b,c
    \end{pmatrix}8^b 7^c.\label{eqn:coset-count-multinomial}
\end{align}
We now specialise to the case where $p=k$. This restriction simplifies the analysis and is sufficient for our purposes. This specific choice was made based on numerical evidence, for small $k$ it was shown that $p=k$ lead to a sufficiently tight bound, it may be that our results could be improved by considering different values of $p$, see table~\ref{tab:table-of-exact-values} for the optimizing value of $p$ for $k\leq 30$.

Since we're attempting to obtain a total Hamming weight of $k$ and we have exactly $k$ blocks each block has to provide (on average) a Hamming weight of $1$. A coset of type $(a,b,c)$ has $a$ blocks providing at least a string of Hamming weight of $0$, $b$ blocks providing at least a string of Hamming weight of $1$ and $c$ blocks providing at least a string of Hamming weight $2$. We therefore restrict our attention to the cosets for which $a=c$. Because of the ``at least'' in the previous statement this is a strict subset of the total cosets, but this subset is  sufficiently large for our bound. To be concrete, for example for $k=12$ we keep the coset $(4,4,4)$ but lose cosets such as $(9,0,3)$ where now $k$ is written as $k = 9\times 0 + 0\times 1 + 3\times 4$ . As all the terms in the sum~\eqref{eqn:coset-count-multinomial} are positive this restricted sum is a lower bound
\begin{align}
    N_k&=\sum_{\substack{a+b \leq k\\k\in a A + bB+cC}} \begin{pmatrix}
        k\\a,b,c
    \end{pmatrix}8^b 7^c\\& \geq \sum_{2a+b =k} \begin{pmatrix}
        k\\a,b,a
    \end{pmatrix}8^b 7^a =: S_k,
\end{align}
where the condition $k\in a A + bB+cC$ no longer appears in the sum, since we have $2a+b = k$, this condition is automatically satisfied.
We can re-express $S_k$ as
\begin{align}
    S_k &= \sum_{a=0}^{\lfloor k/2\rfloor} \begin{pmatrix}
        k\\a,k-2a,a
    \end{pmatrix}8^{k-2a} 7^a,
\end{align}
and write $8=1+7$ and apply the binomial theorem to obtain
\begin{align}
    S_k &= \sum_{a=0}^{\lfloor k/2\rfloor} \sum_{m=0}^{k-2a} \begin{pmatrix}
        k\\a,k-2a,a
    \end{pmatrix}  \begin{pmatrix}k-2a\\m\end{pmatrix} 7^{a+m}.\label{eqn:before-triangle-relabelling-Sk}
\end{align}
By introducing $j=a+m$ and reordering the points in the $(a,m)$-plane we are summing over we can rewrite the sum as
\begin{align}
    S_k = \sum_{j=0}^k 7^j \sum_{a=0}^{\min(j,k-j)}\begin{pmatrix}
        k\\a,k-2a,a
    \end{pmatrix}  \begin{pmatrix}k-2a\\j-a\end{pmatrix},\label{eqn:triangle-relabelled-Sk}
\end{align}
see figure~\ref{fig:triangle-relabelling} for intuition about this reordering.
\begin{figure}[ht!]
\centering
\includegraphics[width=\linewidth]{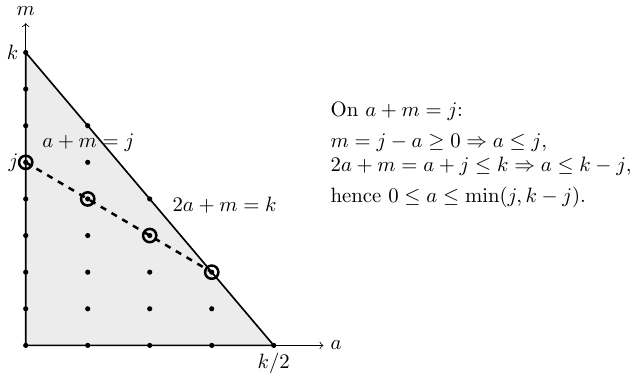}
\caption{Intuition behind the relabelling showing the equivalence of equations~\eqref{eqn:before-triangle-relabelling-Sk} and~\eqref{eqn:triangle-relabelled-Sk}. We switch from summing over $m$ and $a$ to summing over lines of points parallel to the dashed line, and points within those lines.}\label{fig:triangle-relabelling}
\end{figure}

We now restrict our attention to the inner sum in equation~\eqref{eqn:triangle-relabelled-Sk}.
\begin{lem}\label{lem:triangle-combinatorial-lemma}
\begin{align}
    \sum_{a=0}^{\min(j,k-j)}\begin{pmatrix}
        k\\a,k-2a,a
    \end{pmatrix}  \begin{pmatrix}k-2a\\j-a\end{pmatrix} &= \begin{pmatrix}
        k\\j
    \end{pmatrix}^2
    \end{align}
\end{lem}
\begin{proof}
First observe the square of the binomial on the right hand side of equation is exactly the number of ordered pairs $(A,B)$ of subsets of a set of $k$ elements where both subsets have size $j$ (the pairs may overlap, we just take $2$ arbitrary pairs of the correct size).

Now, by expanding the binomial and multinomial in terms of factorials we have
\begin{align}
    \sum_{a=0}^{\min(j,k-j)}&\begin{pmatrix}
        k\\a,k-2a,a
    \end{pmatrix}  \begin{pmatrix}k-2a\\j-a\end{pmatrix}\nonumber\\ &=   \sum_{a=0}^{\min(j,k-j)} \frac{k!}{(a!)^2(j-a)! (k-a-j)!}.
\end{align}
Returning to the ordered pairs of subsets of size $j$ write $a  = |A\setminus B| = | B \setminus A |$. We can then split $[k] = \{1,2,\hdots k\}$ up into $4$ disjoint regions
\begin{align}
    |A\setminus B| &= a\\
    | B \setminus A | &= a\\
    |A \cap B| &= j-a\\
    |[k] \setminus (A \cup B) | &= (k-j-a),
\end{align}
the number of ways to choose these regions is just the multinomial coefficient
\begin{align}
    \begin{pmatrix}
        k\\a,a,(j-a),(k-j-a)
    \end{pmatrix} &= \frac{k!}{(a!)^2(j-a)! (k-a-j)!},
\end{align}
now summing over all the different sizes the intersection of $A$ and $B$ may have proves the result.
\end{proof}

We now have
\begin{align}
    S_k = \sum_{j=0}^k7^j \begin{pmatrix}
        k\\j
    \end{pmatrix}^2,
\end{align}
by solving the equation
\begin{align}
    7 = \frac{x+1}{x-1},
\end{align}
we recognise $S_k$ as proportional to a Legendre polynomial
\begin{align}
    S_k = 6^k P_k\left(\frac{4}{3}\right).
\end{align}
While various asymptotics for the Legendre polynomials are known, in order to obtain an elementary lower bound we use Cauchy-Schwarz with the sequence 
\begin{align}
    x_j = 7^{j/2} \begin{pmatrix}
        k\\j
    \end{pmatrix},
\end{align}
and the constant sequence $y_j=1$
We have
\begin{align}
    (x\cdot y)^2 &\leq (x \cdot x) (y \cdot y)\\
    \left(\sum_{j=0}^k 7^{j/2} \begin{pmatrix}
        k\\j
    \end{pmatrix} \right)^2 &\leq (k+1) S_k,
\end{align}
but from the binomial theorem
\begin{align}
    \sum_{j=0}^k 7^{j/2} \begin{pmatrix}
        k\\j
    \end{pmatrix} = (1+\sqrt{7})^k,
\end{align}
so in summary
\begin{align}
    N_k \geq S_k \geq \frac{1}{1+k}(1+\sqrt{7})^{2k}
\end{align}

\FloatBarrier
\subsection{Denominator}
\label{sec:denominator}
In this section we aim to upper bound the quantity $\rank{C_p(\ket{0})}$. We can express this quantity, for $p\leq n$ as
\begin{align}
    \rank{C_p(\ket{0})} &= \sum_{j=0}^{\lfloor p/2\rfloor } \begin{pmatrix}
        n\\p-2j
    \end{pmatrix},\label{eqn:exact-formula-rank-vacuum}
\end{align}
for $p > n$ we have $\rank{C_p(\ket{0})} = \rank{C_{2n-p}(\ket{0})} $ due fermionic linear optics respecting the overall parity operator, which may be written as $P \propto \prod_{j=1}^{2n} m_j$. To obtain eqation~\eqref{eqn:exact-formula-rank-vacuum} recall that for each fermionic mode we have two Majorana operators: $m_{2j-1} = -i(a_j^\dagger - a_j)$ and $m_{2j} = a_j^\dagger + a_j$. Consider a Majorana monomial of degree $p$, if no Majoranas in the monomial share the same associated mode then this monomial will map the vacuum state to a Fock state with $p$ fermions, the number of ways to do this corresponds to the $j=0$ term in the sum in~\eqref{eqn:exact-formula-rank-vacuum}. Then $j=1$ corresponds to having a single ``collision'' with a pair of Majoranas on the same mode etc. The subspaces corresponding to different numbers of excitations are orthogonal so the dimensions sum. Recalling that we are interested in $4$-mode magic states we have $n=4k$ and in order to match the bound proved on the numerator in section~\ref{sec:numerator} we set $p=k$ so
\begin{align}
    \rank{C_k(\ket{0})} &= \sum_{j=0}^{\lfloor k/2\rfloor } \begin{pmatrix}
        4k\\k-2j
    \end{pmatrix},
\end{align}
now compare the $j^\text{th}$ term in this sum with the $(j+1)^\text{th}$ 
\begin{align}
    \frac{\begin{pmatrix}
        4k\\k-2j
    \end{pmatrix}}{\begin{pmatrix}
        4k\\k-2j -2
    \end{pmatrix}} &= \frac{(k-2j -2)! (3k+2j+2)!}{(k-2j)! (3k+2j)!}\\
    &= \frac{(3k+2j+2)(3k+2j+1)}{(k-2j)(k-2j-1)}\\
    &\geq 9.
\end{align}
From which it follows that we can bound the sum by a geometric series with ratio $1/9$
\begin{align}
    \rank{C_k(\ket{0})} &\leq \begin{pmatrix}
        4k\\k
    \end{pmatrix}\sum_{j=0}^{\lfloor k/2\rfloor}\frac{1}{9^j}\\
    &\leq \frac{9}{8} \begin{pmatrix}
        4k\\k
    \end{pmatrix} .
\end{align}
Now recall that for any $x > 0$ we have
\begin{align}
    (1+x)^{4k} = \sum_{j=0}^{4k} \begin{pmatrix}
        4k\\j
    \end{pmatrix} x^{j} \geq \begin{pmatrix}
        4k\\k
    \end{pmatrix} x^{k},
\end{align}
so we can upper bound
\begin{align}
    \begin{pmatrix}
        4k\\k
    \end{pmatrix} \leq \left( \frac{(1+x)^4}{x}\right)^k.
\end{align}
Differentiating shows that 
\begin{align}
    \frac{(1+x)^4}{x},
\end{align}
is minimized at $x=1/3$ at which point it takes the value $3(4/3)^4 = 4^4/3^3$, therefore
\begin{align}
    \rank{C_k(\ket{0})} \leq \frac{9}{8} \left(\frac{256}{27}\right)^k
\end{align}

\section{Approximate fermionic Gaussian rank}
\label{sec:approximate-rank}
The $\delta$ approximate Gaussian rank, denoted $\chi_\delta(\psi)$  was introduced in~\cite{Dias_2024} following the introduction of the analogous quantity for stabilizer decompositions in~\cite{bravyi-2019} as the minimum $r$ such that there exists a $\psi^\prime$ of rank $r$ with
\begin{align}
    \norm{\psi - \psi^\prime} < \delta.
\end{align}
Since, in general, completely exact classical simulation is unnecessary (and even impossible if one has to resort to floating point arithmetic), it is natural to allow a small controlled error in the Gaussian (or stabilizer) decompositions we use for simulation algorithms. As such the $\delta$-approximate rank is, in many cases, a more natural quantity to define the difficulty of classical simulation.

Although proof will be based singular values of the $C_p(\ket{\phi})$, matrices, to simplify the argument it is convenient to move to slightly different maps. Specifically we define the map $\tilde{C}_p$ which, while defined for any state, is specifically tuned for the stabilizer state $\ket{M}^{\otimes k}$. In section~\ref{sec:numerator} we show how the group of (projective) Majorana monomials may be divided into cosets under the action of the stabilizer group of $\ket{M}^{\otimes k}$, specifically any $Q$ and $P$ (with appropriate Majorana Hamming-weight) such that there exists a stabilizer $g$ satisfying $Q=Pg$ lead to columns of the $C_p(\ket{M}^{\otimes k})$ which are multiples of each other. Our definition of $\tilde{C}_p(\ket{\phi})$ will be that instead of the columns enumerating every Majorana monomial with Hamming-weight $p$, where there are cosets containing multiple representatives of weight $p$ we only choose one of them, and do not include the others. As such $\tilde{C}_p(\ket{\phi})$ is clearly a column-submatrix of the original $C_p(\ket{\phi})$.

Since we do not include all monomials of weight $p$ the map $\tilde{C}_p$ is not covariant under the group of Gaussian unitaries as $C_p$ is. This makes it more difficult to evaluate on arbitrary Gaussian states, however we have the bound
\begin{align}
    \rank{\tilde{C}_p(\ket{f})} \leq \rank{C_p(\ket{f})} =\rank{ C_p(\ket{0})},\label{eqn:inequality-free-state-submatrix-rank-upper-bound}
\end{align}
as passing to a submatrix can not increase the rank. Since we will only need upper bounds on $\tilde{C}_p$ for the analysis, the inequality~\eqref{eqn:inequality-free-state-submatrix-rank-upper-bound} suffices for our purpose.

Assume we have a $\delta$-approximation of the target state, 
\begin{align}
    \norm{\ket{M}^{\otimes k} - \ket{\phi}} \leq \delta \leq 1,
\end{align}
here and in the remainder of this section all vector norms will be the $2$-norm and all matrix norms will be the Frobenius norm. Let $\ket{\phi}$ be expressed as a superposition of $r$ Gaussian states implying
\begin{align}
    \rank{\tilde{C}_p(\ket{\phi})} \leq r \rank{C_p(\ket{0})},
\end{align}
however we have
\begin{align}
    \norm{\tilde{C}_p(\ket{M}^{\otimes k}) - \tilde{C}_p(\ket{\phi})}^2 &= \sum_{i=1}^{N_p} \norm{X_i(\ket{M}^{\otimes k} - \ket{\phi})}^2\\
    &= \sum_{i=1}^{N_p}\norm{\ket{M}^{\otimes k} - \ket{\phi}}^2\\
    &\leq N_p\delta^2, \label{eqn:frobenius-distance-delta-bound}
\end{align}
where $N_p$ is the number of columns of $\tilde{C}_p$, and implicitly depends also on $k$, the number of copies of the magic state, since we have a different $\tilde{C}_p$ for each dimension $n=4k$. In the above equation $X_i$ denotes the weight-$p$ Majorana monomial indexing column $i$ of $\tilde{C}_p$.

Since we have at most one monomial from each coset of the stabilizer group of $\ket{M}^{\otimes k}$ the columns of $\tilde{C}_p(\ket{M}^{\otimes k})$ are orthogonal and hence 
\begin{align}
    \tilde{C}_p(\ket{M}^{\otimes k})^\dagger \tilde{C}_p(\ket{M}^{\otimes k}) = I_{N_p},
\end{align}
we therefore have that $\tilde{C}_p(\ket{M}^{\otimes k})$ has all $N_p$ singular values equal to $1$.

Now by the Eckart–Young–Mirsky theorem (see e.g. Ref.~\cite{horn2012matrix}) on low-rank Frobenius-norm approximation the best approximation of rank $r \rank{C_p(\ket{0})}$ to a matrix with $N_p$ singular values equal to $1$ has squared Frobenius distance at least $N_p - r \rank{C_P(\ket{0})}$, combining this with equation~\eqref{eqn:frobenius-distance-delta-bound} we obtain 
\begin{align}
    N_p - r \rank{C_p(\ket{0})} \leq N_p\delta^2,
\end{align}
implying 
\begin{align}
    (1-\delta^2) \frac{N_p}{\rank{C_p(\ket{0})}} \leq r.
\end{align}

\section{The exact fermonic rank for \texorpdfstring{$k=2$}{k=2} copies}
\label{sec:two-copies-only}
Although theorems~\ref{thm:exact-rank} and~\ref{thm:approx-rank} prove appealing asymptotic statements, they are unfortunately not able to prove that the rank of a tensor product of at least $2$ non-Gaussian parity eigenstates is $4$. The strongest possible result from the above methods is in the third row of table~\ref{tab:table-of-exact-values}, showing the result $r \geq 3$.  In this section we will address this and prove that a tensor product of $k=2$ magic states can not be written as a linear combination of less than $4$ fermionic-Gaussian states. 
We will rely heavily on the results of Refs.~\cite{Holweck-Oeding-2025,Holweck-Levay-2015}. There, the authors studied polynomial $\mathrm{Spin}(16, \mathbb{C})$-invariants of the positive-parity (even) part of the $8$ mode Fock-space, which we will denote $\mathcal{F}_{8}^{+}$. For our purposes this is sufficient since, while no longer unitary, the elements of $\mathrm{Spin}(16, \mathbb{C})$ are invertible and map fermionic Gaussian states to fermionic Gaussian states. In particular therefore, $\mathrm{Spin}(16, \mathbb{C})$ operations do not change the fermionic Gaussian rank of a state. As with the real case, $\mathrm{Spin}(16, \mathbb{C})$ preserves the positive parity subspace $\mathcal{F}_8^+$.

The invariants $f_{2}$, $f_{8}$, $f_{12}$, $f_{14}$, $f_{18}$,$f_{20}$, $f_{24}$, and $f_{30}$ are described in appendix A. of Ref.~\cite{Holweck-Oeding-2025} in terms of a subspace of $\mathcal{F}_{8}^{+}$ which they associate with a Cartan subalgebra related to a decomposition of the exceptional Lie algebra $\mathfrak{e}_8$ as $\mathfrak{e}_8 = \mathfrak{so}(16, \mathbb{C}) \oplus \mathcal{F}^+_8$. For our purposes we will describe this subspace explicitly as the span of the $8$ basis vectors
\begin{align}
\lvert E_1\rangle
&=
\left(
a_1^\dagger a_3^\dagger a_5^\dagger a_7^\dagger
+
a_2^\dagger a_4^\dagger a_6^\dagger a_8^\dagger
\right)\lvert\Omega\rangle \\&=
\lvert10101010\rangle+\lvert01010101\rangle\\
\lvert E_2\rangle
&=\left(
a_1^\dagger a_3^\dagger a_6^\dagger a_8^\dagger
+
a_2^\dagger a_4^\dagger a_5^\dagger a_7^\dagger
\right)\lvert\Omega\rangle \\&=
\lvert10100101\rangle+\lvert01011010\rangle\\
\lvert E_3\rangle
&=
\left(
a_1^\dagger a_4^\dagger a_5^\dagger a_8^\dagger
+
a_2^\dagger a_3^\dagger a_6^\dagger a_7^\dagger
\right)\lvert\Omega\rangle \\&=
\lvert10011001\rangle+\lvert01100110\rangle
\\
\lvert E_4\rangle
&=
\left(
a_1^\dagger a_4^\dagger a_6^\dagger a_7^\dagger
+
a_2^\dagger a_3^\dagger a_5^\dagger a_8^\dagger
\right)\lvert\Omega\rangle \\&=
\lvert10010110\rangle+\lvert01101001\rangle\\
\lvert E_5\rangle
&=
\left(
a_1^\dagger a_2^\dagger a_7^\dagger a_8^\dagger
+
a_3^\dagger a_4^\dagger a_5^\dagger a_6^\dagger
\right)\lvert\Omega\rangle\\&=
\lvert11000011\rangle+\lvert00111100\rangle\\
\lvert E_6\rangle
&=
\left(
a_1^\dagger a_2^\dagger a_5^\dagger a_6^\dagger
+
a_3^\dagger a_4^\dagger a_7^\dagger a_8^\dagger
\right)\lvert\Omega\rangle\\&=
\lvert11001100\rangle+\lvert00110011\rangle\\
\lvert E_7\rangle
&=
\left(
a_1^\dagger a_2^\dagger a_3^\dagger a_4^\dagger
+
a_5^\dagger a_6^\dagger a_7^\dagger a_8^\dagger
\right)\lvert\Omega\rangle \\&=
\lvert11110000\rangle+\lvert00001111\rangle\\
\lvert E_8\rangle
&=
\left(
\mathbf{1}
+
a_1^\dagger a_2^\dagger a_3^\dagger a_4^\dagger
a_5^\dagger a_6^\dagger a_7^\dagger a_8^\dagger
\right)\lvert\Omega\rangle\\&=
\lvert00000000\rangle+\lvert11111111\rangle.
\end{align}
Here we have departed from the notation of Ref.~\cite{Holweck-Oeding-2025} in order to better match the quantum information literature. The relation between our $a_i^\dagger$ creation operators and their $p_{i}$ and $p_{\bar{i}}$ is given by $a_{2i}^\dagger = p_{\bar{i}}$ and $a_{2i-1}^\dagger = p_{i}$ for $1\leq i \leq 4$. While it is not the case that every vector in $\mathcal{F}^+_8$ is complex-spin equivalent to a vector in the span of the $\ket{E_i}$, this is the case for every so-called \emph{semi-simple} vector. For our purposes this is sufficient because the semi-simple elements are dense, so by a continuity argument our proof for semi-simple elements is sufficient. 

Our argument proceeds via a concrete parameterization of the space formed by linear combinations of (up to) three fermionic Gaussian states. As proved by Manivel~\cite{Manivel-2009} the expression
\begin{align}
    \ket{\alpha, \beta,\gamma} = \alpha \ket{0}^{\otimes 8} + \beta\ket{1}^{\otimes 8} + \gamma \left(\ket{00} + \ket{11}\right)^{\otimes 4},\label{eqn:manival-normal-form}
\end{align}
is a normal-form for general points in the second secant variety, more specifically the orbit of the set 
\begin{align}
    \{\ket{\alpha, \beta,\gamma} \,|\, (\alpha,\beta,\gamma) \in \mathbb{C}^3\},
\end{align}
under the group $\mathrm{Spin}(16, \mathbb{C})$ is open and dense in the second secant variety. We now show that any $\ket{\alpha,\beta,\gamma}$ can be brought to a state in the Cartan subspace spanned by the $\ket{E_i}$ by a $\mathrm{Spin}(16,\mathbb{C})$ transformation. To fix some notation note that if we pair the first qubit with the second, the third with the fourth and so on, the state $\ket{\alpha,\beta,\gamma}$ is invariant under permutations of the $4$, paired-qubit factors, i.e. under the action of the symmetric group which permutes the $4$ tensor factors in the term multiplying $\gamma$ in equation~\eqref{eqn:manival-normal-form}. The subsequent analysis will stay inside the subspace which is symmetric in this way and as a shorthand we will simply call such states symmetric. 

Now given a symmetric state in one can associate a polynomial in variables $X$ and $Y$ according to
\begin{align}
    f_{\psi}(X, Y) = \left(X \bra{00} + Y\bra{11}\right)^{\otimes 4} \ket{\psi},
\end{align}
and, in fact, for symmetric states the association $\psi \leftrightarrow f_\psi$ is injective, one can recover the symmetric state from the polynomial. The normal form $\ket{\alpha,\beta,\gamma}$ is associated with the polynomial
\begin{align}
    \ket{\alpha, \beta,\gamma} \leftrightarrow f_{\alpha,\beta,\gamma}(X, Y)=  \alpha X^4 +\beta Y^4 + \gamma (X+Y)^4.
\end{align}
Now, for a matrix 
\begin{align}
g = \begin{pmatrix}
    a & b\\c&d
\end{pmatrix}
\end{align} 
in $\mathrm{Sl}(2,\mathbb{C})$ we have the action
\begin{align}
    \rho(g) \ket{00} &= a\ket{00} + c\ket{11}\\
    \rho(g) \ket{11} &= b\ket{00} + d\ket{11},
\end{align}
this action respects our association between symmetric states and polynomials since
\begin{align}
    \rho(g)^{\otimes 4} \ket{\psi} \leftrightarrow f_{\psi}(aX+cY, bX+dY),
\end{align}
note that for any $g\in \mathrm{Sl}(2,\mathbb{C})$ the action $\rho(g)^{\otimes 4}\in \mathrm{Spin}(16,\mathbb{C})$. A well-known result due to Cayley (see Ref.~\cite{Xiao-binary-forms}, Lemma 2.3) states that under $\mathrm{GL}(2,\mathbb{C})$ transformations any quartic polynomial in two variables with non-zero discriminant may be brought to the normal form
\begin{align}
    X^4 + A X^2Y^2 + Y^4,
\end{align}
for some complex number $A$. If we restrict the group from $\mathrm{GL}(2,\mathbb{C})$ to $\mathrm{SL}(2,\mathbb{C})$ we lose the ability to set the coefficient of $X^4$ and $Y^4$ to $1$, so we obtain the normal form
\begin{align}
    B(X^4 + Y^4) + A X^2Y^2.\label{eqn:normal-form-quartic-binomial-poly}
\end{align}
The restriction that the discriminant is non-zero restricts us to an open dense subset of the second secant, however since we will ultimately be obtaining a polynomial invariant detecting the second secant, our proof will apply to the entire second secant by continuity.   

We observe, by inverting the map from symmetric states to polynomials, if $g$ is the $\mathrm{SL}(2,\mathbb{C})$ transformation bringing the polynomial to the form~\eqref{eqn:normal-form-quartic-binomial-poly} then
\begin{align}
    \rho(g)^{\otimes 4} \ket{\alpha, \beta,\gamma} &= B(\ket{00}^{\otimes 4} + \ket{11}^{\otimes 4}) \\&\quad+ A \sum_{1\leq i < j\leq 4} a_{2i-1}^\dagger  a_{2i}^\dagger a_{2j-1}^\dagger  a_{2j}^\dagger \ket{0}\nonumber,
\end{align}
where the last term arises because the term $X^2 Y^2$ means we want two of the four pairs of modes to be full, and we sum over the six different ways of choosing these pairs. Now we observe that the above form is in the Cartan subspace as
\begin{align}
    &B(\ket{00}^{\otimes 4} + \ket{11}^{\otimes 4}) + A \sum_{1\leq i < j\leq 4} a_{2i-1}^\dagger  a_{2i}^\dagger a_{2j-1}^\dagger  a_{2j}^\dagger \ket{0}\nonumber\\ &\quad= B\ket{E_8} + A (\ket{E_5} + \ket{E_6} + \ket{E_7}).\label{eqn:normal-form-second-secant-cartan}
\end{align}
In order to prove the separation of $\ket{M}^{\otimes 2}$ from the second secant, we seek a $\mathrm{Spin}(16, \mathbb{C})$ invariant polynomial $P$ such that $P = 0$ on the second secant. By invariance it suffices to obtain a $\mathrm{Spin}(16, \mathbb{C})$ invariant on the Cartan subspace which is zero for states of the form~\eqref{eqn:normal-form-second-secant-cartan}.  Given this, and the expression for the tensor product of two magic states in the Cartan basis as 
\begin{align}
    \ket{M}^{\otimes 2} \propto \ket{E_7} + \ket{E_8}
\end{align}
the problem can now be solved purely mechanically using Gr{\"o}bner bases. By means of the Sympy computer algebra system~\cite{sympy} we obtain two homogeneous $\mathrm{Spin}(16,\mathbb{C})$ invariant polynomials constructed to be zero on the second secant variety, whose restriction to the Cartan subspace may be written as
\begin{align}
    P_1 &= -17809 F_2^7 + 9266400000 F_2^3 F_8 \nonumber\\&\quad- 5101056000000 F_2 F_{12} \\&\quad+ 44789760000000 F_{14}\nonumber\\
    P_2 &=  7301 F_2^8 - 3412800000 F_2^4 F_8\nonumber\\ &\quad+ 1492992000000 F_2^2 F_{12}\\&\quad- 130636800000000 F_8^2, \nonumber
\end{align}
in terms of the $F$ polynomials provided in the appendix of Ref.~\cite{Holweck-Oeding-2025} with the corresponding values on the target state
\begin{align}
    P_1(\ket{E_7} + \ket{E_8}) &= -134719200000000\\
    P_2(\ket{E_7} + \ket{E_8}) &= 10287648000000000.
\end{align}
Here the homogeneity of the polynomials, combined with the construction that they are $0$ on the second secant variety allows us to disregard the overall normalization of the states, since normalising the vector would only multiply each polynomial by an overall constant. The python script that produces these polynomials and verifies that they are identically zero on the second secant variety is available at Ref.~\cite{secant-computation-github}.

\section{Coherent state rank of bosonic Fock states}

In this section we prove a conjecture formulated in Ref.~\cite{cottier2026lowerboundscoherentstate} for bosonic linear optics. The key difference in this setting is that bosons obey commuting, rather than anti-commuting statistics, and in particular do not obey the Pauli-exclusion principle. Whereas in a system of $n$ fermionic modes a mode can contain either $0$ or $1$ fermions, an $n$ mode bosonic system can contain an unbounded number of bosons on each mode. A general bosonic state on $n$ modes can be written in Fock basis as
\begin{align}
    \ket{\psi} &= \sum_{k\in\mathbb{N}^n} \psi_{k} \ket{k},
\end{align}
where $\ket{k}$, for $k = (k_1, k_2,\hdots k_n)$ is the Fock state with $k_i$ bosons in mode $i$. In this basis, a single mode coherent state has the form
\begin{align}
    \ket{\alpha} &= \exp\left(-\frac{\abs{\alpha}^2}{2}\right)\sum_{k=0}^{\infty} \frac{\alpha^k}{\sqrt{k!}} \ket{k},\label{eqn:coherent-state-def}
\end{align}
for $\alpha\in\mathbb{C}$ (where we set $0^0 = 1$). An $n$ mode coherent state is just a mode-wise tensor product of single mode coherent states. If now $\alpha\in \mathbb{C}^n$, $k\in \mathbb{N}^n$ and we define the exponential and factorial for vectors as
\begin{align}
    \alpha^k &= \prod_{j=1}^n \alpha_j^{k_j}\\
    k! &= \prod_{j=1}^n (k_j!),
\end{align}
essentially the same formula as~\eqref{eqn:coherent-state-def} expresses a general $n$ mode coherent state
\begin{align}
    \ket{\alpha} &= \exp\left(-\frac{\norm{\alpha}^2}{2}\right)\sum_{k\in\mathbb{N}^n} \frac{\alpha^k}{\sqrt{k!}} \ket{k},\label{eqn:coherent-state-n-mode-def}
\end{align}
 In this section we will be concerned with the rank (more precisely, the border rank) of Fock states, over the coherent states. I.e. we seek to lower bound the integer $r$ such that we can express a given Fock state as a limit of superpositions of at most $r$ coherent states.  We first restate Conjecture~1 from Ref.~\cite{cottier2026lowerboundscoherentstate}
\begin{con}
    The coherent state border rank of the Fock state $\ket{1,\hdots, 1}^{\otimes n}$ is $2^n$. 
\end{con}
Note that here the distinction between exact rank and border rank is essential. It is well-known that the exact rank of a Fock state over coherent states is infinite~\cite{cottier2026lowerboundscoherentstate}. We answer this conjecture (in the affirmative) with our theorem~\ref{thm:coherent-state-rank-for-fock-states}.
\begin{thm}\label{thm:coherent-state-rank-for-fock-states}
    Let $\ket{\phi_s}$, $s\in\mathbb{N}$ be a sequence of states such that
    \begin{align}
        \lim_{s\to\infty} \ket{\phi_s} = \ket{m_1, m_2,\hdots, m_n},
    \end{align}
    and each $\ket{\phi_s}$ is a superposition of $r$ coherent states, then $r \geq \prod_{j=1}^n(1+m_j)$. Moreover, a sequence $\ket{\phi_s}$ is available with rank matching this lower bound so this is the exact border rank.
\end{thm}
Denoting by $m$ the vector $(m_1,\hdots m_n)$ we introduce a ``flattening'' matrix $M^m$ specialized to each $m$ (here $M^m$ denotes the matrix associated with $m\in \mathbb{N}^n$, not some matrix $M$ to the power $m$). The matrix will be square and indexed by $p,q \in [m_1]\times [m_2]\times\hdots \times[m_n]$ where $[m_j] = \{0, 1,\hdots m_j\}$. The dimensions of $M^m$ are therefore $d_m\times d_m$ where $d_m = \prod_{j=1}^n (1+m_j)$ and its entries are
\begin{align}
    M^m_{pq}(\ket{\psi}) &= \sqrt{(p+q)!}\braket{p+q}{\psi}\\
    &= \sqrt{(p+q)!}\psi_{p+q},
\end{align}
here, although each $p$ and $q$ live inside the product of cyclic groups $[m_1]\times [m_2]\times\hdots \times[m_n]$ the addition $p+q$ in the above equations takes place in $\mathbb{N}^n$, not the group. We note that when $n=1$ this is essentially the Hankel matrix considered by Ref.~\cite{cottier2026lowerboundscoherentstate}.  Two key observations lead to theorem~\ref{thm:coherent-state-rank-for-fock-states}, firstly $M^m$ applied to any coherent state has rank-$1$, and secondly $M^m$ applied to the Fock state $\ket{m} =\ket{m_1,\hdots,m_n} $ has rank $d$. We now show the first of these observations
\begin{lem}\label{lem:m-rank-1-for-coherent-states}
    For any coherent state $\ket{\alpha}$, and any $m$ the matrix $M^m(\ket{\alpha})$ has rank $1$.
\end{lem}
To prove this write
\begin{align}
    M^m(\ket{\alpha})_{pq} &= \sqrt{(p+q)!}\braket{p+q}{\alpha}\\
    &= \sqrt{(p+q)!} \exp\left(-\frac{\norm{\alpha}^2}{2}\right) \frac{\alpha^{p+q}}{\sqrt{(p+q)!}}\\
    &=  \exp\left(-\frac{\norm{\alpha}^2}{2}\right) \alpha^p \alpha^q,
\end{align}
defining $v(\alpha)$ to be vectors indexed by elements of $[m_1]\times [m_2]\times\hdots \times[m_n]$ such that $v(\alpha)_p = \alpha^p$ it follows that $M^m(\alpha)$ is proportional to the outer product of $v(\alpha)$ with itself, and hence has rank $1$.

We formulate the second observation as lemma~\ref{lem:rank-m-m-is-full}.
\begin{lem}\label{lem:rank-m-m-is-full}
    For any $m\in \mathbb{N}^n$ the rank of $M^m(\ket{m})$ is $d = \prod_j(1+m_j)$.
\end{lem}
Recall that $M^m(\phi)$ is a $d\times d$ matrix so lemma~\ref{lem:rank-m-m-is-full} is just saying that $M^m(\ket{m})$ has full rank. To see this we use again the formula for the matrix entries
\begin{align}
    M^m(\ket{m}))_{pq} &= \sqrt{(p+q)!}\braket{p+q}{m},
\end{align}
recall that $p_j \in [m_j]$ means  $0\leq p_j \leq m_j$ implying that $0\leq m_j-p_j \leq m_j$. It follows that for every vector $p$ there is a unique $q$ such that $p+q = m$ and hence each $p$ there is a unique $q$ such that $\braket{p+q}{m}$ and hence $M^m(\ket{m}))_{pq} $ is non-zero. In fact we see that $M^m(\ket{m}))$ is just a permutation matrix multiplied by $\sqrt{m!}$ and hence has full rank.

Now to prove theorem~\ref{thm:coherent-state-rank-for-fock-states} write 
\begin{align}
    \ket{\phi_s} &= \sum_{j=1}^r a_{sj} \ket{\alpha_{sj}},
\end{align}
by linearity of the map $M^m$ and subaditivity of the standard matrix rank we have
\begin{align}
    M^m(\ket{\phi_s}) &= \sum_{j=1}^r a_{sj} M^m(\ket{\alpha_{sj}})\\
    \mathrm{rank}(M^m(\ket{\phi_s})) &\leq\sum_{j=1}^r \mathrm{rank}(M^m(\ket{\alpha_{sj}})) \\
    &= r,\label{eqn:coherent-state-rank-main-bound}
\end{align}
where the last equality follows from lemma~\ref{lem:m-rank-1-for-coherent-states}. We now note that each $M^m$ is a continuous function of its input state (each one is a function of a finite number of Fock-basis coefficients) so
\begin{align}
    \lim_{s\to\infty} M^m(\ket{\phi_s}) = M^m(\ket{m}),
\end{align}
further the set of matrices of rank at most $k$ for any given $k$ is closed in the usual topology, from which we observe that if the rank of each $M^m(\ket{\phi_s})$ is at most $r$, so is the rank of $ M^m(\ket{m})$. It follows by applying lemma~\ref{lem:rank-m-m-is-full} that
\begin{align}
    r \geq M^m(\ket{m}) &= \prod_{j=1}^n(1+m_j),
\end{align}
which is the lower bound of theorem~\ref{thm:coherent-state-rank-for-fock-states}. To prove this lower bound is also an upper bound, and hence also the exact border rank we first restrict our attention to a single mode and introduce unnormalized state-vectors proportional to the single-mode coherent states
\begin{align}
    \ket{E(\alpha)} &= \sum_{k=0}^\infty \frac{\alpha^k}{\sqrt{k!}} \ket{k}, 
\end{align}
which is an entire (holomorphic everywhere) function of the complex variable $\alpha$. We have
\begin{align}
   \left.\frac{d^m}{d\alpha^m} \ket{E(\alpha)} \right\rvert_{\alpha = 0} = \sqrt{m!} \ket{m}.
\end{align}
From the standard finite-differences formula for the $m^\text{th}$ derivative we therefore obtain
\begin{align}
  \frac{1}{\sqrt{m!}} \lim_{\epsilon\to 0} \left[ \frac{1}{\epsilon^m}\sum_{k=0}^m(-1)^{m-k}\begin{pmatrix}
        m\\k
    \end{pmatrix} \ket{E(k\epsilon)}\right] = \ket{m},
\end{align}
this is a rank $m$ superposition of coherent states which in the limit $\epsilon \to 0$ recovers $\ket{m}$. For a multi-mode Fock state one can take the product of the decomposition for each mode to obtain theorem~\ref{thm:coherent-state-rank-for-fock-states}.

\subsection{Robust bounds}
We are also able to use our methods to prove a robust version of theorem~\ref{thm:coherent-state-rank-for-fock-states}. Much like in the fermionic case, our robust bounds are given by a function depending on the approximation multiplying the exact bound, and again this function is monotonic in the degree of approximation and has a limit of $1$ in the case where the approximation becomes exact. In this case the bounds are most conveniently given in terms of the fidelity, which we denote for normalized pure states by
\begin{align}
    F(\phi, \psi) = \norm{\braket{\phi}{\psi}}^2.
\end{align}
We state the main result of this subsection in theorem~\ref{thm:coherent-rank-robust}.
\begin{thm}\label{thm:coherent-rank-robust}
    Let 
    \begin{align}
        F(\ket{\psi}, \ket{m}) = F,
    \end{align}
    and 
    \begin{align}
        \ket{\psi} &= \sum_{j=1}^r a_j\ket{\alpha_j},
    \end{align}
    then 
    \begin{align}
        r \geq d  \left(1 + \frac{(2m!)(1-F)}{d (m!) F}\right)^{-1} = d\omega,
    \end{align}
    where as in theorem~\ref{thm:coherent-state-rank-for-fock-states}, $d = \prod_{j=1}^n(1+m_j)$ and the quantity $\omega$ in the lower bound is monotonically increasing in $F$ and takes value $1$ as $F\to 1$.
\end{thm}
To prove theorem~\ref{thm:coherent-rank-robust} first recall from the proof of lemma~\ref{lem:rank-m-m-is-full} that $M^m(\ket{m})$ is proportional to a permutation matrix, and write $\pi_m$  for this permutation so that
\begin{align}
    M^m(\ket{m}) &= \sqrt{m!} \pi_m
\end{align}
Letting $\langle\cdot, \cdot\rangle_F$ be the Frobenius inner product, and by considering the elements of $M$ that $\pi_m$ permutes to the diagonal we have 
\begin{align}
    \langle \pi_m, M^m(\psi) \rangle_F = d \sqrt{m!} \phi_m,
\end{align}
so
\begin{align}\label{eqn:frob-inner-product}
   \abs{\langle \pi_m, M^m(\psi) \rangle_F} = d \sqrt{m!}  \sqrt{F},
\end{align}
but because $\pi_m$ is a permutation it is unitary and so via Cauchy-Schwarz
\begin{align}\label{eqn:cauchy-schwarz-matrices}
    \abs{\langle \pi_m, M^m(\psi) \rangle_F} \leq \norm{M(\psi)}_1
\end{align}
Also since the matrix rank of $M^m(\psi)$ is not greater than $r$, the Shatten-$1$ and $2$ norms of $M^m(\psi)$ satisfy
\begin{align}\label{eqn:one-norm-bounds-2}
    \norm{M^m(\psi)}_1 \leq \sqrt{r} \norm{M^m(\psi)}_F.
\end{align}
Combining equations~\eqref{eqn:frob-inner-product},~\eqref{eqn:cauchy-schwarz-matrices} and~\eqref{eqn:one-norm-bounds-2}  we obtain
\begin{align}
    d \sqrt{m!}  \sqrt{F} &\leq \sqrt{r} \norm{M^m(\psi)}_F\\
     \frac{d^2 (m!) F}{\norm{M^m(\psi)}_F^2} \leq r.\label{eqn:basic-rank-bound}
\end{align}
We now bound
\begin{align}
    \norm{M^m(\psi)}_F^2 &= \tr(M^m(\psi)M^m(\psi)^\dagger)\\
    &= \sum_{p,q} (p +q)! \abs{\psi_{p+q}}^2\\
    &= \sum_{k\in [2m]} P(k) k! \abs{\psi_{k}}^2,
\end{align}
where $P(k)$ counts how many $(p,q)$ in $([m_1]\times [m_2]\times \hdots[m_n])^2$ sum to $k$ and $[2m]$ denotes the set $[2m_1]\times [2m_2]\times \hdots[2m_n]$. Recalling that we expect $\ket{\psi}$ to be a good approximation of the Fock state $\ket{m}$ we split the sum into the $k=m$ component and everything else. To bound the ``everything else'' term note that $P(k) k! \leq (2m)!$ (this bound is extremely loose but is independent of $k$), for the $k=m$ term we have $P(m) = d$ from the same reasoning as in the proof of lemma~\ref{lem:rank-m-m-is-full}
\begin{align}
    \norm{M(\psi)}_F^2 \leq d m! F + (2m!)(1-F).
\end{align}
Combining this with the above bound and using~\eqref{eqn:basic-rank-bound} we have
\begin{align}
   \frac{d^2 (m!) F}{d (m!) F + (2m!)(1-F)} \leq r,
\end{align}
proving theorem~\ref{thm:coherent-rank-robust}. For $\ket{m} = \ket{1}^{\otimes n}$ we have $d = 2^n$, $m! = 1$ and $2m! = 2^n$ so
\begin{align}
   r\geq \frac{2^{2n} F}{2^n F + 2^n (1-F)}  = 2^{n} F ,
\end{align}
matching the exact bound up to a factor of the fidelity $F$.

\section{Conclusion and outlook}
\label{sec:conclusion}
The rank and approximate rank control, respectively, the ultimate running time for a class of classical simulation algorithms for exact and approximate simulation of quantum circuits. In this work we have proved exponential lower bounds for the fermionic Gaussian rank, border rank and approximate rank of magic states, as well as the coherent state border rank and approximate rank of bosonic Fock states. In the fermionic case there remains a large gap between the lower bounds we prove $\Omega(1.4^k)$ and the known upper bound of $O(2^k)$, and the exact value of the rank within this region is completely unknown for $k > 2$. It may be noteworthy that the strongest lower bound we prove for fixed, low dimensional systems ($k \leq 30$) appears to be very close to $2^{k/2}$. It would be interesting to explore whether one could prove something asymptotic to this as a lower or even upper bound on the rank.

In our fermionic work we focus on tensor products of $4$ qubit states, and reduce the rank of tensor products of larger states to that of this class by lemma~\ref{lem:arbitrary-tensor-product}. It seems likely that these bounds arising from this reduction are extremely loose and would be of interest to obtain stronger bounds by directly considering larger tensor factors. Given the stabilizer flavor of the reasoning we employ for $\ket{M}^{\otimes k}$, it is likely that our results could be applied to an arbitrary stabilizer state, but we are not aware of any non-equivalent examples of stabilizer states where the fermionic Gaussian rank would be computationally or physically relevant.

The overall method we employ is quite general and, indeed similar approaches have already been applied to the study of entanglement~\cite{Steffan-phd-thesis} the map to Hankel matrices employed in the recent work on Bosonic linear optics~\cite{cottier2026lowerboundscoherentstate} appears to be closely related. More broadly, our approach can be viewed as an instance of the rank-method, which has been well explored in algebraic geometry~\cite{equations-for-secant,buczyński2026cactusbarriers} both abstractly and with application to very concrete problems such as matrix multiplication~\cite{doležálek2026nonlinearmethodstensorsdeterminantal} and the computation of the determinant and permanent of matrices~\cite{han2025recursivekoszulflatteningsdeterminant}. Given the generality and flexibility of our methods, an obvious question is whether similar methods can be applied to bound the \emph{stabilizer} rank of (e.g. $\ket{T}$) magic states. This is of significant theoretical and practical interest due to the privileged roles stabiliser and $\ket{T}$ states play in fault tolerant quantum computation. To our knowledge the matrix rank method we have employed here has not been applied to the stabilizer setting. There appear to be significant obstacles to a direct application. Without assuming the set of operators chosen to build the $C$ matrices is preserved by the Clifford group it seems difficult to bound the ``denominator'' term in the rank bound. Then if the set is Clifford-invariant it is clear that it can not include a single non-identity Pauli operator, since it would then include all Pauli operators and both the numerator and denominator of the rank bound become $2^n$, leading to triviality. It is possible that one could apply the method with richer sets of operators, for example by using the recently fully characterized Clifford-commutant~\cite{Bittel2026completetheoryof}. In particular there is no reason, a priori, to demand that the operators chosen map from a given vector space to itself, they may map to some larger space.

Beyond the realm of classical simulation, an important direction we leave open is to understand the implications our results have for the resource theory of fermionic non-Gaussianity. 
While early attempts to quantify fermionic non-Gaussianity date back as far as Ref.~\cite{PhysRevLett.95.123003}, this resource theory has recently received significant interest~\cite{coffman2025measuringnongaussianmagicfermions,lyu2024fermionicgaussiantestingnongaussian,haug2026practicaltestswitnessesfermionic,turkeshi2026williamsonmajorizationtheoryfermionic,tarabunga2026computablefermionicnongaussianitycovariance,tarabunga2026fermionicnongaussianitybellsampling}, from researchers bridging the gap between condensed matter and quantum information theory. In this context fermionic linear optics acts as a ``mean-field theory'', upon which perturbations (in the guise of non-Gassianity, or interactions) may be applied. Non-Gaussianity has been observed driving phase transitions~\cite{3yx4-1j27,tarabunga2026computablefermionicnongaussianitycovariance}, and controlling the cost of learning unknown, close to Gaussian quantum states~\cite{PRXQuantum.6.010319}. 

\textbf{Acknowledgments:}
I am grateful to Hakop Pashayan and Micha{\l} Oszmaniec for helpful discussions, and also to the latter for providing the productive working environment within which this work was done. I acknowledge the support from National Science Center, Poland within the QuantERA III Programme (No 2023/05/Y/ST2/00140 acronym Tuquan). This work was performed independently of the interesting Ref.~\cite{wei2026robustexponentiallowerbounds}, appearing in a previous arXiv listing.

\textbf{AI statement:} ChatGPT 5.6 Sol was used for literature review and some calculations, especially the combinatorics in subsection~\ref{sec:numerator} where it formulated and proved lemma~\ref{lem:triangle-combinatorial-lemma} and ``noticed'', for want of a more appropriate term, the Legendre polynomial appearing in section~\ref{sec:numerator}. It was also used to generalize the finite-difference method giving the upper bound in theorem~\ref{thm:coherent-state-rank-for-fock-states} from $\ket{1}^{\otimes n}$ to general Fock states, produce Figure~\ref{fig:triangle-relabelling} and a first version of the Gr{\"o}bner basis script available at Ref.~\cite{secant-computation-github}, but not the script used to produce figure~\ref{fig:exact-rank-bounds} and table~\ref{tab:table-of-exact-values}. All the calculations have been checked by the author and remain the authors sole responsibility. LLMs were not used for any of the text, which the author feels he is fully capable of messing up himself.

\bibliography{refs}

\begin{thebibliography}{54}%
\makeatletter
\providecommand \@ifxundefined [1]{%
 \@ifx{#1\undefined}
}%
\providecommand \@ifnum [1]{%
 \ifnum #1\expandafter \@firstoftwo
 \else \expandafter \@secondoftwo
 \fi
}%
\providecommand \@ifx [1]{%
 \ifx #1\expandafter \@firstoftwo
 \else \expandafter \@secondoftwo
 \fi
}%
\providecommand \natexlab [1]{#1}%
\providecommand \enquote  [1]{``#1''}%
\providecommand \bibnamefont  [1]{#1}%
\providecommand \bibfnamefont [1]{#1}%
\providecommand \citenamefont [1]{#1}%
\providecommand \href@noop [0]{\@secondoftwo}%
\providecommand \href [0]{\begingroup \@sanitize@url \@href}%
\providecommand \@href[1]{\@@startlink{#1}\@@href}%
\providecommand \@@href[1]{\endgroup#1\@@endlink}%
\providecommand \@sanitize@url [0]{\catcode `\\12\catcode `\$12\catcode
  `\&12\catcode `\#12\catcode `\^12\catcode `\_12\catcode `\%12\relax}%
\providecommand \@@startlink[1]{}%
\providecommand \@@endlink[0]{}%
\providecommand \url  [0]{\begingroup\@sanitize@url \@url }%
\providecommand \@url [1]{\endgroup\@href {#1}{\urlprefix }}%
\providecommand \urlprefix  [0]{URL }%
\providecommand \Eprint [0]{\href }%
\providecommand \doibase [0]{https://doi.org/}%
\providecommand \selectlanguage [0]{\@gobble}%
\providecommand \bibinfo  [0]{\@secondoftwo}%
\providecommand \bibfield  [0]{\@secondoftwo}%
\providecommand \translation [1]{[#1]}%
\providecommand \BibitemOpen [0]{}%
\providecommand \bibitemStop [0]{}%
\providecommand \bibitemNoStop [0]{.\EOS\space}%
\providecommand \EOS [0]{\spacefactor3000\relax}%
\providecommand \BibitemShut  [1]{\csname bibitem#1\endcsname}%
\let\auto@bib@innerbib\@empty
\bibitem [{\citenamefont {Cottier}\ and\ \citenamefont
  {Chabaud}(2026)}]{cottier2026lowerboundscoherentstate}%
  \BibitemOpen
  \bibfield  {author} {\bibinfo {author} {\bibfnamefont {F.}~\bibnamefont
  {Cottier}}\ and\ \bibinfo {author} {\bibfnamefont {U.}~\bibnamefont
  {Chabaud}},\ }\href {https://arxiv.org/abs/2604.00766} {\bibinfo {title}
  {Lower bounds on coherent state rank}} (\bibinfo {year} {2026}),\ \Eprint
  {https://arxiv.org/abs/2604.00766} {arXiv:2604.00766 [quant-ph]} \BibitemShut
  {NoStop}%
\bibitem [{\citenamefont {Bravyi}\ and\ \citenamefont
  {Gosset}(2016)}]{bravyi2016improved}%
  \BibitemOpen
  \bibfield  {author} {\bibinfo {author} {\bibfnamefont {S.}~\bibnamefont
  {Bravyi}}\ and\ \bibinfo {author} {\bibfnamefont {D.}~\bibnamefont
  {Gosset}},\ }\href {https://doi.org/10.1103/PhysRevLett.116.250501}
  {\bibfield  {journal} {\bibinfo  {journal} {Phys. Rev. Lett.}\ }\textbf
  {\bibinfo {volume} {116}},\ \bibinfo {pages} {250501} (\bibinfo {year}
  {2016})}\BibitemShut {NoStop}%
\bibitem [{\citenamefont {Bravyi}\ \emph {et~al.}(2019)\citenamefont {Bravyi},
  \citenamefont {Browne}, \citenamefont {Calpin}, \citenamefont {Campbell},
  \citenamefont {Gosset},\ and\ \citenamefont {Howard}}]{bravyi-2019}%
  \BibitemOpen
  \bibfield  {author} {\bibinfo {author} {\bibfnamefont {S.}~\bibnamefont
  {Bravyi}}, \bibinfo {author} {\bibfnamefont {D.}~\bibnamefont {Browne}},
  \bibinfo {author} {\bibfnamefont {P.}~\bibnamefont {Calpin}}, \bibinfo
  {author} {\bibfnamefont {E.}~\bibnamefont {Campbell}}, \bibinfo {author}
  {\bibfnamefont {D.}~\bibnamefont {Gosset}},\ and\ \bibinfo {author}
  {\bibfnamefont {M.}~\bibnamefont {Howard}},\ }\href
  {https://doi.org/10.22331/q-2019-09-02-181} {\bibfield  {journal} {\bibinfo
  {journal} {Quantum}\ }\textbf {\bibinfo {volume} {3}},\ \bibinfo {pages}
  {181} (\bibinfo {year} {2019})}\BibitemShut {NoStop}%
\bibitem [{\citenamefont {Pashayan}\ \emph {et~al.}(2022)\citenamefont
  {Pashayan}, \citenamefont {Reardon-Smith}, \citenamefont {Korzekwa},\ and\
  \citenamefont {Bartlett}}]{PRXQuantum.3.020361}%
  \BibitemOpen
  \bibfield  {author} {\bibinfo {author} {\bibfnamefont {H.}~\bibnamefont
  {Pashayan}}, \bibinfo {author} {\bibfnamefont {O.}~\bibnamefont
  {Reardon-Smith}}, \bibinfo {author} {\bibfnamefont {K.}~\bibnamefont
  {Korzekwa}},\ and\ \bibinfo {author} {\bibfnamefont {S.~D.}\ \bibnamefont
  {Bartlett}},\ }\href {https://doi.org/10.1103/PRXQuantum.3.020361} {\bibfield
   {journal} {\bibinfo  {journal} {PRX Quantum}\ }\textbf {\bibinfo {volume}
  {3}},\ \bibinfo {pages} {020361} (\bibinfo {year} {2022})}\BibitemShut
  {NoStop}%
\bibitem [{\citenamefont {Kissinger}\ \emph {et~al.}(2022)\citenamefont
  {Kissinger}, \citenamefont {van~de Wetering},\ and\ \citenamefont
  {Vilmart}}]{Kissinger-2022}%
  \BibitemOpen
  \bibfield  {author} {\bibinfo {author} {\bibfnamefont {A.}~\bibnamefont
  {Kissinger}}, \bibinfo {author} {\bibfnamefont {J.}~\bibnamefont {van~de
  Wetering}},\ and\ \bibinfo {author} {\bibfnamefont {R.}~\bibnamefont
  {Vilmart}}\ }(\bibinfo  {publisher} {Schloss Dagstuhl – Leibniz-Zentrum
  für Informatik},\ \bibinfo {year} {2022})\ pp.\ \bibinfo {pages}
  {5:1--5:13}\BibitemShut {NoStop}%
\bibitem [{Note1()}]{Note1}%
  \BibitemOpen
  \bibinfo {note} {Qubit quantum circuits consisting of matchgates are exactly
  equivalent to so-called fermionic linear optics due to the results of~\cite
  {PhysRevA.65.032325} and we will use the two terms
  interchangeably.}\BibitemShut {Stop}%
\bibitem [{\citenamefont {Dias}\ and\ \citenamefont
  {Koenig}(2024)}]{Dias_2024}%
  \BibitemOpen
  \bibfield  {author} {\bibinfo {author} {\bibfnamefont {B.}~\bibnamefont
  {Dias}}\ and\ \bibinfo {author} {\bibfnamefont {R.}~\bibnamefont {Koenig}},\
  }\href {https://doi.org/10.22331/q-2024-05-21-1350} {\bibfield  {journal}
  {\bibinfo  {journal} {Quantum}\ }\textbf {\bibinfo {volume} {8}},\ \bibinfo
  {pages} {1350} (\bibinfo {year} {2024})}\BibitemShut {NoStop}%
\bibitem [{\citenamefont {Reardon-Smith}\ \emph {et~al.}(2024)\citenamefont
  {Reardon-Smith}, \citenamefont {Oszmaniec},\ and\ \citenamefont
  {Korzekwa}}]{ReardonSmith2024improvedsimulation}%
  \BibitemOpen
  \bibfield  {author} {\bibinfo {author} {\bibfnamefont {O.}~\bibnamefont
  {Reardon-Smith}}, \bibinfo {author} {\bibfnamefont {M.}~\bibnamefont
  {Oszmaniec}},\ and\ \bibinfo {author} {\bibfnamefont {K.}~\bibnamefont
  {Korzekwa}},\ }\href {https://doi.org/10.22331/q-2024-12-04-1549} {\bibfield
  {journal} {\bibinfo  {journal} {{Quantum}}\ }\textbf {\bibinfo {volume}
  {8}},\ \bibinfo {pages} {1549} (\bibinfo {year} {2024})}\BibitemShut
  {NoStop}%
\bibitem [{\citenamefont {Gottesman}\ and\ \citenamefont
  {Chuang}(1999)}]{gottesman1999demonstrating}%
  \BibitemOpen
  \bibfield  {author} {\bibinfo {author} {\bibfnamefont {D.}~\bibnamefont
  {Gottesman}}\ and\ \bibinfo {author} {\bibfnamefont {I.~L.}\ \bibnamefont
  {Chuang}},\ }\href {https://doi.org/10.1038/46503} {\bibfield  {journal}
  {\bibinfo  {journal} {Nature}\ }\textbf {\bibinfo {volume} {402}},\ \bibinfo
  {pages} {390} (\bibinfo {year} {1999})}\BibitemShut {NoStop}%
\bibitem [{\citenamefont {Zhou}\ \emph {et~al.}(2000)\citenamefont {Zhou},
  \citenamefont {Leung},\ and\ \citenamefont {Chuang}}]{zhou2000methodology}%
  \BibitemOpen
  \bibfield  {author} {\bibinfo {author} {\bibfnamefont {X.}~\bibnamefont
  {Zhou}}, \bibinfo {author} {\bibfnamefont {D.~W.}\ \bibnamefont {Leung}},\
  and\ \bibinfo {author} {\bibfnamefont {I.~L.}\ \bibnamefont {Chuang}},\
  }\href {https://doi.org/10.1103/PhysRevA.62.052316} {\bibfield  {journal}
  {\bibinfo  {journal} {Phys. Rev. A}\ }\textbf {\bibinfo {volume} {62}},\
  \bibinfo {pages} {052316} (\bibinfo {year} {2000})}\BibitemShut {NoStop}%
\bibitem [{\citenamefont {Hebenstreit}\ \emph {et~al.}(2019)\citenamefont
  {Hebenstreit}, \citenamefont {Jozsa}, \citenamefont {Kraus}, \citenamefont
  {Strelchuk},\ and\ \citenamefont
  {Yoganathan}}]{Yoganathan-2019-magic-for-matchgates}%
  \BibitemOpen
  \bibfield  {author} {\bibinfo {author} {\bibfnamefont {M.}~\bibnamefont
  {Hebenstreit}}, \bibinfo {author} {\bibfnamefont {R.}~\bibnamefont {Jozsa}},
  \bibinfo {author} {\bibfnamefont {B.}~\bibnamefont {Kraus}}, \bibinfo
  {author} {\bibfnamefont {S.}~\bibnamefont {Strelchuk}},\ and\ \bibinfo
  {author} {\bibfnamefont {M.}~\bibnamefont {Yoganathan}},\ }\href
  {https://doi.org/10.1103/PhysRevLett.123.080503} {\bibfield  {journal}
  {\bibinfo  {journal} {Phys. Rev. Lett.}\ }\textbf {\bibinfo {volume} {123}},\
  \bibinfo {pages} {080503} (\bibinfo {year} {2019})}\BibitemShut {NoStop}%
\bibitem [{\citenamefont {Bravyi}(2006)}]{PhysRevA.73.042313}%
  \BibitemOpen
  \bibfield  {author} {\bibinfo {author} {\bibfnamefont {S.}~\bibnamefont
  {Bravyi}},\ }\href {https://doi.org/10.1103/PhysRevA.73.042313} {\bibfield
  {journal} {\bibinfo  {journal} {Phys. Rev. A}\ }\textbf {\bibinfo {volume}
  {73}},\ \bibinfo {pages} {042313} (\bibinfo {year} {2006})}\BibitemShut
  {NoStop}%
\bibitem [{\citenamefont {Dias}\ and\ \citenamefont
  {K\"onig}(2024)}]{PhysRevA.110.042402}%
  \BibitemOpen
  \bibfield  {author} {\bibinfo {author} {\bibfnamefont {B.}~\bibnamefont
  {Dias}}\ and\ \bibinfo {author} {\bibfnamefont {R.}~\bibnamefont {K\"onig}},\
  }\href {https://doi.org/10.1103/PhysRevA.110.042402} {\bibfield  {journal}
  {\bibinfo  {journal} {Phys. Rev. A}\ }\textbf {\bibinfo {volume} {110}},\
  \bibinfo {pages} {042402} (\bibinfo {year} {2024})}\BibitemShut {NoStop}%
\bibitem [{\citenamefont {Hahn}\ \emph {et~al.}(2025)\citenamefont {Hahn},
  \citenamefont {Takagi}, \citenamefont {Ferrini},\ and\ \citenamefont
  {Yamasaki}}]{Hahn2025classicalsimulation}%
  \BibitemOpen
  \bibfield  {author} {\bibinfo {author} {\bibfnamefont {O.}~\bibnamefont
  {Hahn}}, \bibinfo {author} {\bibfnamefont {R.}~\bibnamefont {Takagi}},
  \bibinfo {author} {\bibfnamefont {G.}~\bibnamefont {Ferrini}},\ and\ \bibinfo
  {author} {\bibfnamefont {H.}~\bibnamefont {Yamasaki}},\ }\href
  {https://doi.org/10.22331/q-2025-10-13-1881} {\bibfield  {journal} {\bibinfo
  {journal} {{Quantum}}\ }\textbf {\bibinfo {volume} {9}},\ \bibinfo {pages}
  {1881} (\bibinfo {year} {2025})}\BibitemShut {NoStop}%
\bibitem [{\citenamefont {Mehraban}\ and\ \citenamefont
  {Tahmasbi}(2024)}]{mehraban2024quadraticlowerboundsapproximate}%
  \BibitemOpen
  \bibfield  {author} {\bibinfo {author} {\bibfnamefont {S.}~\bibnamefont
  {Mehraban}}\ and\ \bibinfo {author} {\bibfnamefont {M.}~\bibnamefont
  {Tahmasbi}},\ }\href {https://arxiv.org/abs/2305.10277} {\bibinfo {title}
  {Quadratic lower bounds on the approximate stabilizer rank: A probabilistic
  approach}} (\bibinfo {year} {2024}),\ \Eprint
  {https://arxiv.org/abs/2305.10277} {arXiv:2305.10277 [quant-ph]} \BibitemShut
  {NoStop}%
\bibitem [{Note2()}]{Note2}%
  \BibitemOpen
  \bibinfo {note} {Here, by effective we mean that e.g. high complexity is not
  hidden in the coefficients of the decomposition. A priori one could have even
  a constant rank, but coefficients that are exponentially difficult to
  compute.}\BibitemShut {Stop}%
\bibitem [{\citenamefont {Huang}\ \emph {et~al.}(2018)\citenamefont {Huang},
  \citenamefont {Newman},\ and\ \citenamefont
  {Szegedy}}]{huang2018explicitlowerboundsstrong}%
  \BibitemOpen
  \bibfield  {author} {\bibinfo {author} {\bibfnamefont {C.}~\bibnamefont
  {Huang}}, \bibinfo {author} {\bibfnamefont {M.}~\bibnamefont {Newman}},\ and\
  \bibinfo {author} {\bibfnamefont {M.}~\bibnamefont {Szegedy}},\ }\href
  {https://arxiv.org/abs/1804.10368} {\bibinfo {title} {Explicit lower bounds
  on strong quantum simulation}} (\bibinfo {year} {2018}),\ \Eprint
  {https://arxiv.org/abs/1804.10368} {arXiv:1804.10368 [quant-ph]} \BibitemShut
  {NoStop}%
\bibitem [{\citenamefont {Oszmaniec}\ \emph {et~al.}(2022)\citenamefont
  {Oszmaniec}, \citenamefont {Dangniam}, \citenamefont {Morales},\ and\
  \citenamefont {Zimbor{\'a}s}}]{oszmaniec2022fermion}%
  \BibitemOpen
  \bibfield  {author} {\bibinfo {author} {\bibfnamefont {M.}~\bibnamefont
  {Oszmaniec}}, \bibinfo {author} {\bibfnamefont {N.}~\bibnamefont {Dangniam}},
  \bibinfo {author} {\bibfnamefont {M.~E.}\ \bibnamefont {Morales}},\ and\
  \bibinfo {author} {\bibfnamefont {Z.}~\bibnamefont {Zimbor{\'a}s}},\
  }\href@noop {} {\bibfield  {journal} {\bibinfo  {journal} {PRX Quantum}\
  }\textbf {\bibinfo {volume} {3}},\ \bibinfo {pages} {020328} (\bibinfo {year}
  {2022})}\BibitemShut {NoStop}%
\bibitem [{\citenamefont {Hebenstreit}\ \emph {et~al.}(2020)\citenamefont
  {Hebenstreit}, \citenamefont {Jozsa}, \citenamefont {Kraus},\ and\
  \citenamefont {Strelchuk}}]{hebenstreit2020computational}%
  \BibitemOpen
  \bibfield  {author} {\bibinfo {author} {\bibfnamefont {M.}~\bibnamefont
  {Hebenstreit}}, \bibinfo {author} {\bibfnamefont {R.}~\bibnamefont {Jozsa}},
  \bibinfo {author} {\bibfnamefont {B.}~\bibnamefont {Kraus}},\ and\ \bibinfo
  {author} {\bibfnamefont {S.}~\bibnamefont {Strelchuk}},\ }\href@noop {}
  {\bibfield  {journal} {\bibinfo  {journal} {Physical Review A}\ }\textbf
  {\bibinfo {volume} {102}},\ \bibinfo {pages} {052604} (\bibinfo {year}
  {2020})}\BibitemShut {NoStop}%
\bibitem [{\citenamefont {Bravyi}\ \emph {et~al.}(2016)\citenamefont {Bravyi},
  \citenamefont {Smith},\ and\ \citenamefont {Smolin}}]{bravyi2016trading}%
  \BibitemOpen
  \bibfield  {author} {\bibinfo {author} {\bibfnamefont {S.}~\bibnamefont
  {Bravyi}}, \bibinfo {author} {\bibfnamefont {G.}~\bibnamefont {Smith}},\ and\
  \bibinfo {author} {\bibfnamefont {J.~A.}\ \bibnamefont {Smolin}},\ }\href
  {https://doi.org/10.1103/PhysRevX.6.021043} {\bibfield  {journal} {\bibinfo
  {journal} {Phys. Rev. X}\ }\textbf {\bibinfo {volume} {6}},\ \bibinfo {pages}
  {021043} (\bibinfo {year} {2016})}\BibitemShut {NoStop}%
\bibitem [{\citenamefont {Qassim}\ \emph {et~al.}(2021)\citenamefont {Qassim},
  \citenamefont {Pashayan},\ and\ \citenamefont
  {Gosset}}]{Qassim2021improvedupperbounds}%
  \BibitemOpen
  \bibfield  {author} {\bibinfo {author} {\bibfnamefont {H.}~\bibnamefont
  {Qassim}}, \bibinfo {author} {\bibfnamefont {H.}~\bibnamefont {Pashayan}},\
  and\ \bibinfo {author} {\bibfnamefont {D.}~\bibnamefont {Gosset}},\ }\href
  {https://doi.org/10.22331/q-2021-12-20-606} {\bibfield  {journal} {\bibinfo
  {journal} {{Quantum}}\ }\textbf {\bibinfo {volume} {5}},\ \bibinfo {pages}
  {606} (\bibinfo {year} {2021})}\BibitemShut {NoStop}%
\bibitem [{\citenamefont {Peleg}\ \emph {et~al.}(2022)\citenamefont {Peleg},
  \citenamefont {Shpilka},\ and\ \citenamefont {Volk}}]{Peleg2022lowerbounds}%
  \BibitemOpen
  \bibfield  {author} {\bibinfo {author} {\bibfnamefont {S.}~\bibnamefont
  {Peleg}}, \bibinfo {author} {\bibfnamefont {A.}~\bibnamefont {Shpilka}},\
  and\ \bibinfo {author} {\bibfnamefont {B.~L.}\ \bibnamefont {Volk}},\ }\href
  {https://doi.org/10.22331/q-2022-02-15-652} {\bibfield  {journal} {\bibinfo
  {journal} {{Quantum}}\ }\textbf {\bibinfo {volume} {6}},\ \bibinfo {pages}
  {652} (\bibinfo {year} {2022})}\BibitemShut {NoStop}%
\bibitem [{\citenamefont {Lovitz}\ and\ \citenamefont
  {Steffan}(2022)}]{Lovitz2022newtechniques}%
  \BibitemOpen
  \bibfield  {author} {\bibinfo {author} {\bibfnamefont {B.}~\bibnamefont
  {Lovitz}}\ and\ \bibinfo {author} {\bibfnamefont {V.}~\bibnamefont
  {Steffan}},\ }\href {https://doi.org/10.22331/q-2022-04-20-692} {\bibfield
  {journal} {\bibinfo  {journal} {{Quantum}}\ }\textbf {\bibinfo {volume}
  {6}},\ \bibinfo {pages} {692} (\bibinfo {year} {2022})}\BibitemShut {NoStop}%
\bibitem [{\citenamefont {Cudby}\ and\ \citenamefont
  {Strelchuk}(2025)}]{cudby2025gaussiandecompositionmagicstates}%
  \BibitemOpen
  \bibfield  {author} {\bibinfo {author} {\bibfnamefont {J.}~\bibnamefont
  {Cudby}}\ and\ \bibinfo {author} {\bibfnamefont {S.}~\bibnamefont
  {Strelchuk}},\ }\href {https://arxiv.org/abs/2307.12654} {\bibinfo {title}
  {Gaussian decomposition of magic states for matchgate computations}}
  (\bibinfo {year} {2025}),\ \Eprint {https://arxiv.org/abs/2307.12654}
  {arXiv:2307.12654 [quant-ph]} \BibitemShut {NoStop}%
\bibitem [{\citenamefont {Jozsa}\ and\ \citenamefont
  {Miyake}(2008)}]{josaMatchgate2008}%
  \BibitemOpen
  \bibfield  {author} {\bibinfo {author} {\bibfnamefont {R.}~\bibnamefont
  {Jozsa}}\ and\ \bibinfo {author} {\bibfnamefont {A.}~\bibnamefont {Miyake}},\
  }\href {https://doi.org/10.1098/rspa.2008.0189} {\bibfield  {journal}
  {\bibinfo  {journal} {Proc. R. Soc. A.}\ }\textbf {\bibinfo {volume} {464}},\
  \bibinfo {pages} {3089–3106} (\bibinfo {year} {2008})}\BibitemShut
  {NoStop}%
\bibitem [{\citenamefont {Reardon-Smith}(2026)}]{secant-computation-github}%
  \BibitemOpen
  \bibfield  {author} {\bibinfo {author} {\bibfnamefont {O.}~\bibnamefont
  {Reardon-Smith}},\ }\href@noop {} {\bibinfo {title}
  {Gaussian-rank-magic-states}},\ \bibinfo {howpublished}
  {\url{https://github.com/or1426/Gaussian-rank-magic-states}} (\bibinfo {year}
  {2026})\BibitemShut {NoStop}%
\bibitem [{\citenamefont {Wei}\ \emph {et~al.}(2026)\citenamefont {Wei},
  \citenamefont {Wu}, \citenamefont {Liu},\ and\ \citenamefont
  {Liu}}]{wei2026robustexponentiallowerbounds}%
  \BibitemOpen
  \bibfield  {author} {\bibinfo {author} {\bibfnamefont {F.}~\bibnamefont
  {Wei}}, \bibinfo {author} {\bibfnamefont {K.-W.}\ \bibnamefont {Wu}},
  \bibinfo {author} {\bibfnamefont {Z.}~\bibnamefont {Liu}},\ and\ \bibinfo
  {author} {\bibfnamefont {Z.-W.}\ \bibnamefont {Liu}},\ }\href
  {https://arxiv.org/abs/2610.02172} {\bibinfo {title} {Robust exponential
  lower bounds for fermionic and bosonic gaussian ranks}} (\bibinfo {year}
  {2026}),\ \Eprint {https://arxiv.org/abs/2610.02172} {arXiv:2610.02172
  [quant-ph]} \BibitemShut {NoStop}%
\bibitem [{\citenamefont {Galgano}(2026)}]{Galgano-secant-varieties}%
  \BibitemOpen
  \bibfield  {author} {\bibinfo {author} {\bibfnamefont {V.}~\bibnamefont
  {Galgano}},\ }\href {https://doi.org/10.1142/S0219199725500890} {\bibfield
  {journal} {\bibinfo  {journal} {Communications in Contemporary Mathematics}\
  }\textbf {\bibinfo {volume} {28}},\ \bibinfo {pages} {2550089} (\bibinfo
  {year} {2026})}\BibitemShut {NoStop}%
\bibitem [{\citenamefont {Landsberg}\ and\ \citenamefont
  {Ottaviani}(2013{\natexlab{a}})}]{Landsberg-equations-for-secant-varieties}%
  \BibitemOpen
  \bibfield  {author} {\bibinfo {author} {\bibfnamefont {J.~M.}\ \bibnamefont
  {Landsberg}}\ and\ \bibinfo {author} {\bibfnamefont {G.}~\bibnamefont
  {Ottaviani}},\ }\bibfield  {journal} {\bibinfo  {journal} {Annali di
  Matematica Pura ed Applicata}\ }\textbf {\bibinfo {volume} {192}},\ \href
  {https://doi.org/10.1007/s10231-011-0238-6} {10.1007/s10231-011-0238-6}
  (\bibinfo {year} {2013}{\natexlab{a}})\BibitemShut {NoStop}%
\bibitem [{\citenamefont {Holweck}\ and\ \citenamefont
  {Oeding}(2025)}]{Holweck-Oeding-2025}%
  \BibitemOpen
  \bibfield  {author} {\bibinfo {author} {\bibfnamefont {F.}~\bibnamefont
  {Holweck}}\ and\ \bibinfo {author} {\bibfnamefont {L.}~\bibnamefont
  {Oeding}},\ }\bibfield  {journal} {\bibinfo  {journal} {Ann. Inst. Henri
  Poincaré Comb. Phys. Interact.}\ }\textbf {\bibinfo {volume} {12}},\ \href
  {https://doi.org/10.4171/AIHPD/187} {10.4171/AIHPD/187} (\bibinfo {year}
  {2025})\BibitemShut {NoStop}%
\bibitem [{\citenamefont {L\'evay}\ and\ \citenamefont
  {Holweck}(2015)}]{Holweck-Levay-2015}%
  \BibitemOpen
  \bibfield  {author} {\bibinfo {author} {\bibfnamefont {P.}~\bibnamefont
  {L\'evay}}\ and\ \bibinfo {author} {\bibfnamefont {F.}~\bibnamefont
  {Holweck}},\ }\href {https://doi.org/10.1103/PhysRevD.91.125029} {\bibfield
  {journal} {\bibinfo  {journal} {Phys. Rev. D}\ }\textbf {\bibinfo {volume}
  {91}},\ \bibinfo {pages} {125029} (\bibinfo {year} {2015})}\BibitemShut
  {NoStop}%
\bibitem [{\citenamefont {Manivel}(2009)}]{Manivel-2009}%
  \BibitemOpen
  \bibfield  {author} {\bibinfo {author} {\bibfnamefont {L.}~\bibnamefont
  {Manivel}},\ }\bibfield  {journal} {\bibinfo  {journal} {Symmetry,
  Integrability and Geometry: Methods and Applications}\ }\href
  {https://doi.org/10.3842/SIGMA.2009.078} {10.3842/SIGMA.2009.078} (\bibinfo
  {year} {2009})\BibitemShut {NoStop}%
\bibitem [{\citenamefont {Xiao}(2019)}]{Xiao-binary-forms}%
  \BibitemOpen
  \bibfield  {author} {\bibinfo {author} {\bibfnamefont {S.~Y.}\ \bibnamefont
  {Xiao}},\ }\href {https://doi.org/10.5802/jtnb.1083} {\bibfield  {journal}
  {\bibinfo  {journal} {Journal de th\'eorie des nombres de Bordeaux}\ }\textbf
  {\bibinfo {volume} {31}},\ \bibinfo {pages} {323} (\bibinfo {year}
  {2019})}\BibitemShut {NoStop}%
\bibitem [{\citenamefont {Oszmaniec}\ \emph {et~al.}(2014)\citenamefont
  {Oszmaniec}, \citenamefont {Gutt},\ and\ \citenamefont
  {Ku\ifmmode~\acute{s}\else \'{s}\fi{}}}]{Oszmaniec2014}%
  \BibitemOpen
  \bibfield  {author} {\bibinfo {author} {\bibfnamefont {M.}~\bibnamefont
  {Oszmaniec}}, \bibinfo {author} {\bibfnamefont {J.}~\bibnamefont {Gutt}},\
  and\ \bibinfo {author} {\bibfnamefont {M.}~\bibnamefont
  {Ku\ifmmode~\acute{s}\else \'{s}\fi{}}},\ }\href
  {https://doi.org/10.1103/PhysRevA.90.020302} {\bibfield  {journal} {\bibinfo
  {journal} {Phys. Rev. A}\ }\textbf {\bibinfo {volume} {90}},\ \bibinfo
  {pages} {020302} (\bibinfo {year} {2014})}\BibitemShut {NoStop}%
\bibitem [{\citenamefont {Bravyi}(2005)}]{bravyi-2005-lagrangian-rep}%
  \BibitemOpen
  \bibfield  {author} {\bibinfo {author} {\bibfnamefont {S.}~\bibnamefont
  {Bravyi}},\ }\href {https://doi.org/10.48550/arXiv.quant-ph/0404180}
  {\bibfield  {journal} {\bibinfo  {journal} {Quantum Inf. Comput.}\ }\textbf
  {\bibinfo {volume} {5}},\ \bibinfo {pages} {216–238} (\bibinfo {year}
  {2005})}\BibitemShut {NoStop}%
\bibitem [{\citenamefont
  {Tarabunga}(2026)}]{tarabunga2026fermionicnongaussianitybellsampling}%
  \BibitemOpen
  \bibfield  {author} {\bibinfo {author} {\bibfnamefont {P.~S.}\ \bibnamefont
  {Tarabunga}},\ }\href {https://doi.org/10.48550/arXiv.2606.05066} {\bibinfo
  {title} {Fermionic non-gaussianity via bell sampling: monotones and efficient
  quantum algorithms}} (\bibinfo {year} {2026}),\ \Eprint
  {https://arxiv.org/abs/2606.05066} {arXiv:2606.05066 [quant-ph]} \BibitemShut
  {NoStop}%
\bibitem [{\citenamefont {Landsberg}\ and\ \citenamefont
  {Ottaviani}(2013{\natexlab{b}})}]{equations-for-secant}%
  \BibitemOpen
  \bibfield  {author} {\bibinfo {author} {\bibfnamefont {J.~M.}\ \bibnamefont
  {Landsberg}}\ and\ \bibinfo {author} {\bibfnamefont {G.}~\bibnamefont
  {Ottaviani}},\ }\bibfield  {journal} {\bibinfo  {journal} {Annali di
  Matematica Pura ed Applicata}\ }\textbf {\bibinfo {volume} {192}},\ \href
  {https://doi.org/10.1007/s10231-011-0238-6} {10.1007/s10231-011-0238-6}
  (\bibinfo {year} {2013}{\natexlab{b}})\BibitemShut {NoStop}%
\bibitem [{\citenamefont {Buczyński}(2026)}]{buczyński2026cactusbarriers}%
  \BibitemOpen
  \bibfield  {author} {\bibinfo {author} {\bibfnamefont {J.}~\bibnamefont
  {Buczyński}},\ }\href {https://arxiv.org/abs/2602.11309} {\bibinfo {title}
  {Cactus barriers}} (\bibinfo {year} {2026}),\ \Eprint
  {https://arxiv.org/abs/2602.11309} {arXiv:2602.11309 [math.AG]} \BibitemShut
  {NoStop}%
\bibitem [{\citenamefont {Steffan}(2023)}]{Steffan-phd-thesis}%
  \BibitemOpen
  \bibfield  {author} {\bibinfo {author} {\bibfnamefont {V.}~\bibnamefont
  {Steffan}},\ }\emph {\bibinfo {title} {Tensor Decompositions: Theory and
  Applications in Quantum Information}},\ \href@noop {} {Ph.D. thesis},\
  \bibinfo  {school} {School of the Faculty of Science, University of
  Copenhagen} (\bibinfo {year} {2023})\BibitemShut {NoStop}%
\bibitem [{\citenamefont {Christandl}\ \emph {et~al.}(2018)\citenamefont
  {Christandl}, \citenamefont {Jensen},\ and\ \citenamefont
  {Zuiddam}}]{CHRISTANDL2018125}%
  \BibitemOpen
  \bibfield  {author} {\bibinfo {author} {\bibfnamefont {M.}~\bibnamefont
  {Christandl}}, \bibinfo {author} {\bibfnamefont {A.~K.}\ \bibnamefont
  {Jensen}},\ and\ \bibinfo {author} {\bibfnamefont {J.}~\bibnamefont
  {Zuiddam}},\ }\href
  {https://doi.org/https://doi.org/10.1016/j.laa.2017.12.020} {\bibfield
  {journal} {\bibinfo  {journal} {Linear Algebra and its Applications}\
  }\textbf {\bibinfo {volume} {543}},\ \bibinfo {pages} {125} (\bibinfo {year}
  {2018})}\BibitemShut {NoStop}%
\bibitem [{\citenamefont {Horn}\ and\ \citenamefont
  {Johnson}(2012)}]{horn2012matrix}%
  \BibitemOpen
  \bibfield  {author} {\bibinfo {author} {\bibfnamefont {R.~A.}\ \bibnamefont
  {Horn}}\ and\ \bibinfo {author} {\bibfnamefont {C.~R.}\ \bibnamefont
  {Johnson}},\ }\href {https://doi.org/10.1017/CBO9780511810817} {\emph
  {\bibinfo {title} {Matrix analysis}}}\ (\bibinfo  {publisher} {Cambridge
  university press},\ \bibinfo {year} {2012})\BibitemShut {NoStop}%
\bibitem [{\citenamefont {Meurer}\ \emph {et~al.}(2017)\citenamefont {Meurer}
  \emph {et~al.}}]{sympy}%
  \BibitemOpen
  \bibfield  {author} {\bibinfo {author} {\bibfnamefont {A.}~\bibnamefont
  {Meurer}} \emph {et~al.},\ }\href {https://doi.org/10.7717/peerj-cs.103}
  {\bibfield  {journal} {\bibinfo  {journal} {PeerJ Computer Science}\ }\textbf
  {\bibinfo {volume} {3}},\ \bibinfo {pages} {103} (\bibinfo {year}
  {2017})}\BibitemShut {NoStop}%
\bibitem [{\citenamefont {Doležálek}\ and\ \citenamefont
  {Michałek}(2026)}]{doležálek2026nonlinearmethodstensorsdeterminantal}%
  \BibitemOpen
  \bibfield  {author} {\bibinfo {author} {\bibfnamefont {M.}~\bibnamefont
  {Doležálek}}\ and\ \bibinfo {author} {\bibfnamefont {M.}~\bibnamefont
  {Michałek}},\ }\href {https://arxiv.org/abs/2602.12762} {\bibinfo {title}
  {Nonlinear methods for tensors: determinantal equations for secant varieties
  beyond cactus}} (\bibinfo {year} {2026}),\ \Eprint
  {https://arxiv.org/abs/2602.12762} {arXiv:2602.12762 [math.AG]} \BibitemShut
  {NoStop}%
\bibitem [{\citenamefont {Han}\ \emph {et~al.}(2025)\citenamefont {Han},
  \citenamefont {Ju},\ and\ \citenamefont
  {Kim}}]{han2025recursivekoszulflatteningsdeterminant}%
  \BibitemOpen
  \bibfield  {author} {\bibinfo {author} {\bibfnamefont {J.~I.}\ \bibnamefont
  {Han}}, \bibinfo {author} {\bibfnamefont {J.-H.}\ \bibnamefont {Ju}},\ and\
  \bibinfo {author} {\bibfnamefont {Y.}~\bibnamefont {Kim}},\ }\href
  {https://arxiv.org/abs/2503.12032} {\bibinfo {title} {Recursive koszul
  flattenings of determinant and permanent tensors}} (\bibinfo {year} {2025}),\
  \Eprint {https://arxiv.org/abs/2503.12032} {arXiv:2503.12032 [math.AC]}
  \BibitemShut {NoStop}%
\bibitem [{\citenamefont {Bittel}\ \emph {et~al.}(2026)\citenamefont {Bittel},
  \citenamefont {Eisert}, \citenamefont {Leone}, \citenamefont {Mele},\ and\
  \citenamefont {Oliviero}}]{Bittel2026completetheoryof}%
  \BibitemOpen
  \bibfield  {author} {\bibinfo {author} {\bibfnamefont {L.}~\bibnamefont
  {Bittel}}, \bibinfo {author} {\bibfnamefont {J.}~\bibnamefont {Eisert}},
  \bibinfo {author} {\bibfnamefont {L.}~\bibnamefont {Leone}}, \bibinfo
  {author} {\bibfnamefont {A.~A.}\ \bibnamefont {Mele}},\ and\ \bibinfo
  {author} {\bibfnamefont {S.~F.}\ \bibnamefont {Oliviero}},\ }\href
  {https://doi.org/10.22331/q-2026-07-22-2171} {\bibfield  {journal} {\bibinfo
  {journal} {{Quantum}}\ }\textbf {\bibinfo {volume} {10}},\ \bibinfo {pages}
  {2171} (\bibinfo {year} {2026})}\BibitemShut {NoStop}%
\bibitem [{\citenamefont {Gottlieb}\ and\ \citenamefont
  {Mauser}(2005)}]{PhysRevLett.95.123003}%
  \BibitemOpen
  \bibfield  {author} {\bibinfo {author} {\bibfnamefont {A.~D.}\ \bibnamefont
  {Gottlieb}}\ and\ \bibinfo {author} {\bibfnamefont {N.~J.}\ \bibnamefont
  {Mauser}},\ }\href {https://doi.org/10.1103/PhysRevLett.95.123003} {\bibfield
   {journal} {\bibinfo  {journal} {Phys. Rev. Lett.}\ }\textbf {\bibinfo
  {volume} {95}},\ \bibinfo {pages} {123003} (\bibinfo {year}
  {2005})}\BibitemShut {NoStop}%
\bibitem [{\citenamefont {Coffman}\ \emph {et~al.}(2025)\citenamefont
  {Coffman}, \citenamefont {Smith},\ and\ \citenamefont
  {Gao}}]{coffman2025measuringnongaussianmagicfermions}%
  \BibitemOpen
  \bibfield  {author} {\bibinfo {author} {\bibfnamefont {L.}~\bibnamefont
  {Coffman}}, \bibinfo {author} {\bibfnamefont {G.}~\bibnamefont {Smith}},\
  and\ \bibinfo {author} {\bibfnamefont {X.}~\bibnamefont {Gao}},\ }\href
  {https://arxiv.org/abs/2501.06179} {\bibinfo {title} {Measuring non-gaussian
  magic in fermions: Convolution, entropy, and the violation of wick's theorem
  and the matchgate identity}} (\bibinfo {year} {2025}),\ \Eprint
  {https://arxiv.org/abs/2501.06179} {arXiv:2501.06179 [quant-ph]} \BibitemShut
  {NoStop}%
\bibitem [{\citenamefont {Lyu}\ and\ \citenamefont
  {Bu}(2024)}]{lyu2024fermionicgaussiantestingnongaussian}%
  \BibitemOpen
  \bibfield  {author} {\bibinfo {author} {\bibfnamefont {X.}~\bibnamefont
  {Lyu}}\ and\ \bibinfo {author} {\bibfnamefont {K.}~\bibnamefont {Bu}},\
  }\href {https://arxiv.org/abs/2409.08180} {\bibinfo {title} {Fermionic
  gaussian testing and non-gaussian measures via convolution}} (\bibinfo {year}
  {2024}),\ \Eprint {https://arxiv.org/abs/2409.08180} {arXiv:2409.08180
  [quant-ph]} \BibitemShut {NoStop}%
\bibitem [{\citenamefont {Haug}\ \emph {et~al.}(2026)\citenamefont {Haug},
  \citenamefont {Turkeshi},\ and\ \citenamefont
  {Sierant}}]{haug2026practicaltestswitnessesfermionic}%
  \BibitemOpen
  \bibfield  {author} {\bibinfo {author} {\bibfnamefont {T.}~\bibnamefont
  {Haug}}, \bibinfo {author} {\bibfnamefont {X.}~\bibnamefont {Turkeshi}},\
  and\ \bibinfo {author} {\bibfnamefont {P.}~\bibnamefont {Sierant}},\ }\href
  {https://doi.org/10.48550/arXiv.2605.26218} {\bibinfo {title} {Practical
  tests and witnesses of fermionic non-gaussianity}} (\bibinfo {year} {2026}),\
  \Eprint {https://arxiv.org/abs/2605.26218} {arXiv:2605.26218 [quant-ph]}
  \BibitemShut {NoStop}%
\bibitem [{\citenamefont {Turkeshi}\ \emph {et~al.}(2026)\citenamefont
  {Turkeshi}, \citenamefont {Sierant},\ and\ \citenamefont
  {Tarabunga}}]{turkeshi2026williamsonmajorizationtheoryfermionic}%
  \BibitemOpen
  \bibfield  {author} {\bibinfo {author} {\bibfnamefont {X.}~\bibnamefont
  {Turkeshi}}, \bibinfo {author} {\bibfnamefont {P.}~\bibnamefont {Sierant}},\
  and\ \bibinfo {author} {\bibfnamefont {P.~S.}\ \bibnamefont {Tarabunga}},\
  }\href {https://doi.org/10.48550/arXiv.2608.10140} {\bibinfo {title}
  {Williamson majorization theory of fermionic non-gaussianity}} (\bibinfo
  {year} {2026}),\ \Eprint {https://arxiv.org/abs/2608.10140} {arXiv:2608.10140
  [quant-ph]} \BibitemShut {NoStop}%
\bibitem [{\citenamefont {Tarabunga}\ \emph {et~al.}(2026)\citenamefont
  {Tarabunga}, \citenamefont {Jobst}, \citenamefont {Morral-Yepes},
  \citenamefont {Langer}, \citenamefont {Kraus}, \citenamefont {Pollmann},\
  and\ \citenamefont
  {Lin}}]{tarabunga2026computablefermionicnongaussianitycovariance}%
  \BibitemOpen
  \bibfield  {author} {\bibinfo {author} {\bibfnamefont {P.~S.}\ \bibnamefont
  {Tarabunga}}, \bibinfo {author} {\bibfnamefont {B.}~\bibnamefont {Jobst}},
  \bibinfo {author} {\bibfnamefont {R.}~\bibnamefont {Morral-Yepes}}, \bibinfo
  {author} {\bibfnamefont {M.}~\bibnamefont {Langer}}, \bibinfo {author}
  {\bibfnamefont {B.}~\bibnamefont {Kraus}}, \bibinfo {author} {\bibfnamefont
  {F.}~\bibnamefont {Pollmann}},\ and\ \bibinfo {author} {\bibfnamefont
  {S.-H.}\ \bibnamefont {Lin}},\ }\href
  {https://doi.org/10.48550/arXiv.2607.02242} {\bibinfo {title} {Computable
  fermionic non-gaussianity from the covariance matrix}} (\bibinfo {year}
  {2026}),\ \Eprint {https://arxiv.org/abs/2607.02242} {arXiv:2607.02242
  [quant-ph]} \BibitemShut {NoStop}%
\bibitem [{\citenamefont {Sierant}\ \emph {et~al.}(2026)\citenamefont
  {Sierant}, \citenamefont {Stornati},\ and\ \citenamefont
  {Turkeshi}}]{3yx4-1j27}%
  \BibitemOpen
  \bibfield  {author} {\bibinfo {author} {\bibfnamefont {P.}~\bibnamefont
  {Sierant}}, \bibinfo {author} {\bibfnamefont {P.}~\bibnamefont {Stornati}},\
  and\ \bibinfo {author} {\bibfnamefont {X.}~\bibnamefont {Turkeshi}},\ }\href
  {https://doi.org/10.1103/3yx4-1j27} {\bibfield  {journal} {\bibinfo
  {journal} {PRX Quantum}\ }\textbf {\bibinfo {volume} {7}},\ \bibinfo {pages}
  {010302} (\bibinfo {year} {2026})}\BibitemShut {NoStop}%
\bibitem [{\citenamefont {Mele}\ and\ \citenamefont
  {Herasymenko}(2025)}]{PRXQuantum.6.010319}%
  \BibitemOpen
  \bibfield  {author} {\bibinfo {author} {\bibfnamefont {A.~A.}\ \bibnamefont
  {Mele}}\ and\ \bibinfo {author} {\bibfnamefont {Y.}~\bibnamefont
  {Herasymenko}},\ }\href {https://doi.org/10.1103/PRXQuantum.6.010319}
  {\bibfield  {journal} {\bibinfo  {journal} {PRX Quantum}\ }\textbf {\bibinfo
  {volume} {6}},\ \bibinfo {pages} {010319} (\bibinfo {year}
  {2025})}\BibitemShut {NoStop}%
\bibitem [{\citenamefont {Terhal}\ and\ \citenamefont
  {DiVincenzo}(2002)}]{PhysRevA.65.032325}%
  \BibitemOpen
  \bibfield  {author} {\bibinfo {author} {\bibfnamefont {B.~M.}\ \bibnamefont
  {Terhal}}\ and\ \bibinfo {author} {\bibfnamefont {D.~P.}\ \bibnamefont
  {DiVincenzo}},\ }\href {https://doi.org/10.1103/PhysRevA.65.032325}
  {\bibfield  {journal} {\bibinfo  {journal} {Phys. Rev. A}\ }\textbf {\bibinfo
  {volume} {65}},\ \bibinfo {pages} {032325} (\bibinfo {year}
  {2002})}\BibitemShut {NoStop}%
\end{thebibliography}%
\appendix

\section{Table of exact rank bounds for \texorpdfstring{$k\leq 30$}{k ≤ 30}}
\FloatBarrier
 
 \begin{table}[h]
     \centering
\begin{tabular}{c|ccc}
 k  &      bound from $p=k$  &      $p^*$  &      bound from $p^*$  \\\hline
  0  &              1  &                 0  &                          1  \\
  1  &              2  &                 1  &                          2  \\
  2  &              3  &                 2  &                          3  \\
  3  &              4  &                 3  &                          4  \\
  4  &              6  &                 4  &                          6  \\
  5  &              8  &                 5  &                          8  \\
  6  &             11  &                 6  &                         11  \\
  7  &             15  &                 7  &                         15  \\
  8  &             21  &                 7  &                         21  \\
  9  &             29  &                 8  &                         30  \\
 10  &             41  &                 9  &                         42  \\
 11  &             57  &                10  &                         59  \\
 12  &             80  &                11  &                         83  \\
 13  &            112  &                12  &                        117  \\
 14  &            157  &                13  &                        164  \\
 15  &            220  &                14  &                        232  \\
 16  &            308  &                15  &                        326  \\
 17  &            432  &                16  &                        459  \\
 18  &            606  &                16  &                        649  \\
 19  &            850  &                17  &                        918  \\
 20  &           1193  &                18  &                       1298  \\
 21  &           1673  &                19  &                       1833  \\
 22  &           2347  &                20  &                       2589  \\
 23  &           3292  &                21  &                       3653  \\
 24  &           4618  &                22  &                       5152  \\
 25  &           6478  &                23  &                       7263  \\
 26  &           9087  &                24  &                      10235  \\
 27  &          12745  &                24  &                      14457  \\
 28  &          17876  &                25  &                      20442  \\
 29  &          25072  &                26  &                      28888  \\
 30  &          35164  &                27  &                      40805  \\
\end{tabular}
 \caption{Here $p^*$ is the $p$ which maximizes the lower bound for a given $k$, computed using a script available at~\cite{secant-computation-github}.}
     \label{tab:table-of-exact-values}
 \end{table}
In table~\ref{tab:table-of-exact-values}, the quantity $p^*$ is the $p$ which optimizes the lower bound~\ref{eqn:lower-bound-from-matrix-ranks} for each $k$. In the main text we used the simplification $p=k$, but in fact the true optimal $p^*$ appear to be closer to $\lceil\frac{8}{9}k\rceil$. It is likely that improved upper bounds could be obtained by choosing $p$ to more closely match this pattern.
\end{document}